\RequirePackage{fix-cm}

\documentclass{tJBD2e}      

\usepackage{graphicx}

\usepackage{mathptmx}      \usepackage{lineno}
\usepackage[normalem]{ulem}

\usepackage{amsmath,amsthm,amssymb,epic,epsfig,wrapfig,bm} 
\usepackage[numbers,sort&compress]{natbib}
\usepackage{verbatim,enumitem} 
\usepackage{color} 
\usepackage[linktoc = all]{hyperref}
\hypersetup{
    colorlinks,
    citecolor=black,
    filecolor=black,
    linkcolor=blue,
    urlcolor=black
}
\newcommand{\non}{{\nonumber}}
\newcommand{\p}{\partial}

\newcommand{\G}{{\cal G}}
\newcommand{\R}{{\cal R}_0}

\let\Horig\H
\renewcommand{\H}{{\cal H}}
\newcommand{\F}{{\cal F}}

\newcommand{\E}{{\cal E}}
\newcommand{\J}{{\cal J}}
\newcommand{\bR}{{\mathbb R}}
\newcommand{\one}{\bm{1}}
\newcommand{\SI}{\widetilde{SI}}
\newcommand{\IS}{\widetilde{IS}}
\newcommand{\X}{{\bm{X}}}
\newcommand{\XSI}{\bm{X^{SI}}}
\newcommand{\XSS}{\bm{X^{SS}}}
\newcommand{\XSIT}{\bm{X^{\SI}}}
\newcommand{\XSST}{\bm{X^{\SS}}}

\makeatletter
\newcommand*\bigcdot{\mathpalette\bigcdot@{.8}}
\newcommand*\bigcdot@[2]{\mathbin{\vcenter{\hbox{\scalebox{#2}{$\m@th#1\bullet$}}}}}
\makeatother

\newcommand{\XSdot}{\bm{X^{S{\bigcdot}}}}
\renewcommand{\SS}{\widetilde{SS}}

\renewcommand{\th}{{\bm{\theta}}}
\newcommand{\Th}{{\bm{\theta}}}
\newcommand{\Xth}{{\bm{\Theta}}}
\newcommand{\xth}{\Theta}

\newcommand{\partdot}[2]{\dfrac{\partial {#1}}{\partial {#2}}}
\renewcommand{\b}{{\bigcdot}}
\newcommand{\bx}{{\bm{x}}}
\newcommand{\bk}{{\bm{k}}}

\newcommand{\norml}{\left | \left |}
\newcommand{\normr}{\right | \right |}
\newcommand{\ve}{\varepsilon}
\newcommand{\red}{} %{\textcolor{red}}
\newcommand{\blue}{\textcolor{black}}

\newcommand{\subclass}[1]{\newline{\bf AMS Subject Classification:} #1}

\newtheorem{theorem}{Theorem}[section]
\newtheorem{corollary}[theorem]{Corollary}

\newtheorem{lemma}[theorem]{Lemma}

\theoremstyle{remark}
\newtheorem{remark}{Remark}

\begin{document}

\title{The large graph limit of a stochastic epidemic model on a dynamic multilayer network \thanks{
This research has been supported in part by the Mathematical Biosciences Institute and the National Science Foundation under grants  DMS-1440386 and RAPID DMS-1513489.}
}

\author{Karly A. Jacobsen, 
        Mark G. Burch,  Joseph H. Tien and  Grzegorz A. Rempa{\l}a$^\ast$\thanks{$\ast$ Corresponding author}          \\ College of Public Health, Department of Mathematics and Mathematical Biosciences Institute, \\
        The Ohio State University\\ Columbus, OH 43210, USA  
}

\date{Received: date / Accepted: date}

\maketitle

\begin{abstract}
We consider a \blue{Markovian} SIR-type (Susceptible $\to$ Infected $\to$ Recovered) stochastic epidemic process with multiple modes of transmission on a contact network.  The network is given by a random graph following a multilayer configuration model where edges in different layers correspond to potentially infectious contacts of different types.  We assume that the graph structure evolves in response to the epidemic via activation or deactivation of edges \blue{of infectious nodes}. We derive a large graph limit theorem that gives a system of ordinary differential equations (ODEs) describing the evolution of quantities of interest, such as the proportions of infected and susceptible vertices, as the number of nodes tends to infinity.  Analysis of the limiting system elucidates how the coupling of edge activation and deactivation to infection status affects disease dynamics, as illustrated by a two-layer network example with edge types corresponding to community and healthcare contacts.  Our theorem extends some earlier results describing the deterministic limit of stochastic SIR processes on static, single-layer configuration model graphs.  We also describe precisely the conditions for equivalence between our limiting ODEs and the systems obtained via pair approximation, which are widely used in the epidemiological and ecological literature to approximate disease dynamics on networks.  
\red{The flexible modeling framework and asymptotic results have potential application to many disease settings including Ebola dynamics in West Africa, which was the original motivation for this study.}

 \subclass{ 92D30; 60G55;  60F}
\end{abstract}

\section{Introduction} 
	\label{sec:intro}

A fundamental issue in disease dynamics is that contact patterns change in response to infection.  This is particularly salient in the study of disease dynamics on contact networks: infected individuals curtail contacts with their regular community due to illness (e.g. being too sick to go to school or work) but increase their contacts with other segments of the population, such as healthcare workers or caretakers in the home.  The recent Ebola outbreak in West Africa provides a stark example.  The
array and severity of symptoms, including high fever, diarrhea, vomiting, and hemorrhaging, make symptomatic individuals too ill to engage their regular community contacts and, instead, cause individuals to seek care in the home, hospital, or other facility.  This coupling of evolution of network structure to disease status is basic, but a theoretical understanding of how this affects disease dynamics is currently lacking.  

Disease dynamics on networks has been an extremely active area of research in the past 20 years, typically within the SIR-type (Susceptible $\to$ Infected $\to$ Recovered) modeling framework \cite{brockmann2013,dodds2004,House2011,meyers2005,newman2002,Pellis2015,watts1998}.  This has been stimulated in part by the explosion of data on networks of various sorts \cite{balcan2009,bengtsson2015,brockmann2010,colizza2006,easley2010,newman2010,tatem2009,ugander2012,wesolowski2012} and the recognition that network structure can have a dramatic impact on disease dynamics.  Theoretical findings include applications of percolation theory to static networks \cite{newman2002,Grassberger1983}.  Less theory has been developed for networks that change over time, with much work in this area focusing on concurrent partnerships forming and breaking independent of disease status \cite{Altmann1995,Altmann1998,leung2012,leung2015,volz2007,Bansal2010,Eames2004}.  The study of adaptive networks, where the contact structure changes in response to disease progression, is an emerging area, as reviewed by Funk et al. \cite{Funk2010}.  One popular approach is to assume that susceptible individuals break connections to avoid infection \cite{Gross2006, Epstein2008,Shaw2008,Zanette2008}. 
Related works examine behavioral changes due to awareness of infection \cite{Funk2010,Granell2013}.  As these studies indicate, evolving network structure may lead to rich dynamics that are of practical importance for disease forecasting and evaluating public health interventions \cite{Shai2013,Gross2006,Shaw2008}.  

A challenge for understanding disease dynamics on networks is their high dimensionality -- for example, a modest-sized network or graph (here, and elsewhere in the paper, we use these terms interchangeably) may have tens of thousands of nodes and over a hundred thousand edges.  Various approaches have been developed for deriving simpler models to approximate the full network dynamics including grouping vertices by degree \cite{barthelemy2004,barthelemy2005,pastor-satorras2001} or ``effective degree" \cite{lindquist2011,ma2013,Ball2008} and considering stationary degree distributions for dynamic graphs \cite{Altmann1998,leung2012,leung2015}. Two approaches particularly relevant to this work are the {\it pairwise} approach (e.g. early work includes \cite{Rand1999,Keeling1999,keeling1999b}) and the {\it edge-based} approach of Volz and Miller that is applicable to graphs with a specified degree distribution  \cite{Volz2008,Miller2011,Miller2013}.  Both approaches naturally lead to consideration of the disease dynamics in the large graph limit, i.e. when the number of nodes tends to infinity.  Whereas Volz and Miller derived their results heuristically, recent mathematical work has rigorously shown that a deterministic edge-based system of equations is the large graph limit  of an SIR continuous-time Markov process on a static random graph \cite{Decreusefond2012,Janson2014}.    

Multilayer networks, which allow for more complex disease dynamics, have also received much attention recently, as reviewed by Kivel{\" a} et al. \cite{kivela2014}.  In particular, multilayer networks where the interconnected layers can represent different populations have been considered \cite{Wang2011,Rombach2014,Zhao2014,Buono2014,Yagan2013}.  The effect of degree correlation on two-layer networks has also been studied \cite{Shai2012} with each layer being  an 
Erd\Horig{o}s--R{\'e}nyi or Barab{\'a}si-Albert random graph.  Other models involving two-layer graphs were also considered where one layer corresponded to information-spreading and the second to disease transmission \cite{Jo2006,Funk2009,Funk2010b,Granell2013} or where two competing pathogens spread on the two layers \cite{Wei2012,Sahneh2013}.  Multilayer networks have also recently been employed to model temporal networks as sequences of static networks \cite{Valdano2015,Valdano2015b}. The particular class of multilayer networks studied here are those in which the set of nodes is identical in each layer (i.e. node-aligned \cite{kivela2014}) with distinct edge types corresponding to each layer.  These are often referred to as multiplex (or multi-relational) networks \cite{De-Domenico2013,Cardillo2013,Yagan2012,Battiston2014} and can be represented as a graph with edges colored according to type.
\begin{figure}\begin{center}
\includegraphics[width=39mm]{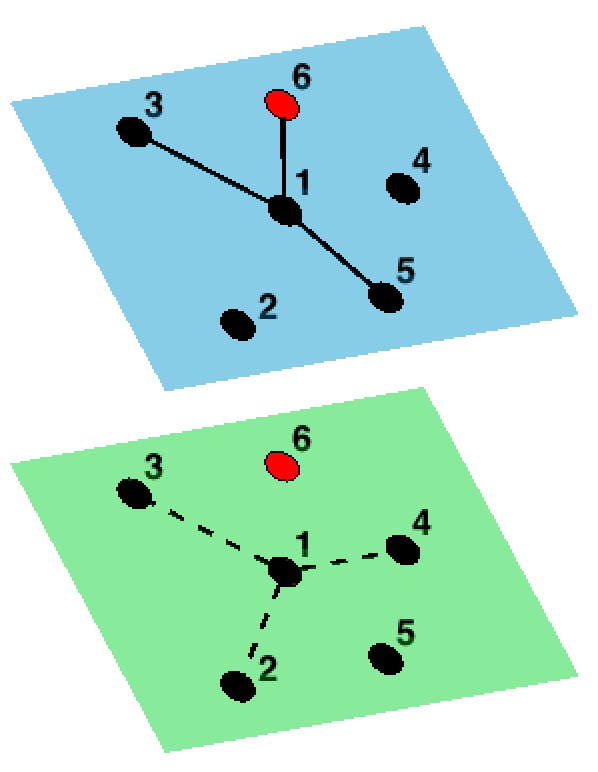}\qquad\qquad  
\includegraphics[width=40mm]{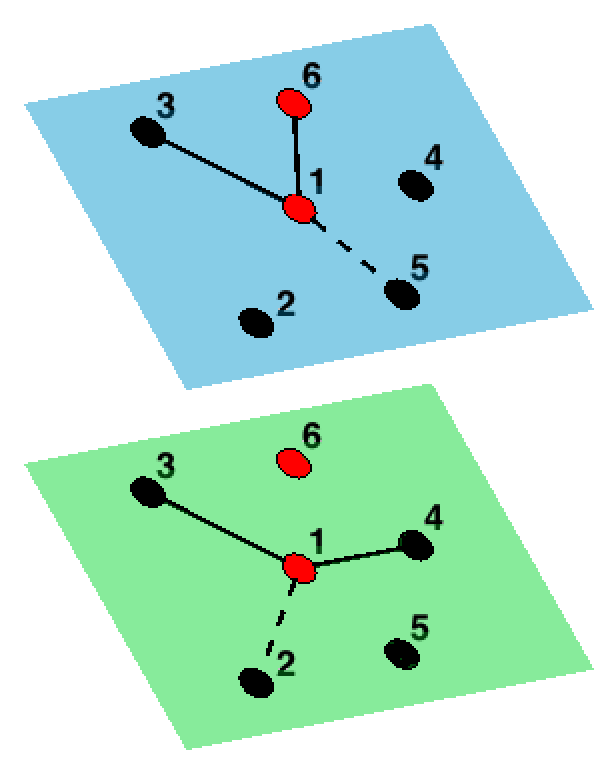}\end{center}

\caption{(Left) Neighborhood of a susceptible vertex (labeled 1) with an infected (red) neighbor.  Community (top/blue) and healthcare (bottom/green) contacts are shown as active (solid) or deactivated (dashed).  (Right) After infection of vertex 1, two of its healthcare contacts are activated and one community contact is dropped.}
\label{fig:nbhd}
\end{figure}

 In this work, we consider the problem of modeling an epidemic with different modes of disease transmission on a dynamic contact network.  Specifically, we formulate an SIR continuous-time stochastic process on a multilayer graph, with specified degree distribution, where nodes represent individuals and edges represent potentially infectious contacts.  Each layer contains the same set of nodes but corresponds to a different transmission mechanism (i.e. a multiplex network).   
 In addition, we allow edges to be active or dormant with transmission occurring only along an active edge. 
 \red{The network structure is dynamic in that edges can activate or deactivate over the course of infection.  This approach allows us to incorporate behavioral changes due to infection while keeping the total  edges (active plus dormant)   given by the degree distribution fixed.}
 A simple example is a two-layer network with one layer corresponding to community contacts and the other to healthcare contacts where we assume that infected individuals deactivate their community edges, for example due to decreased mobility or isolation, while their healthcare edges are being activated as they seek care (Figure \ref{fig:nbhd}).  

The main result of this work, Theorem \ref{thm:lln} in Section \ref{sec:lln}, describes the large graph limit for the stochastic SIR process on the dynamic multilayer network.  According to the theorem, the scaled counts of different edge and node types converge uniformly in probability to the solution of a deterministic system of equations. Thus, we obtain a relatively simple limiting model in the setting where network connectivity changes with the evolution of the disease process.  In particular, it follows that for a certain class of random graphs the large graph limit coincides with the model obtained using either the pairwise \cite{Keeling1999} or the edge-based \cite{volz2007} approximation approach.  As we demonstrate with the two-layer network example, the limiting systems are amenable for mathematical analysis, allowing us to gain biological insight into how changing network structure influences disease dynamics.  Moreover, Theorem \ref{thm:lln} extends previous results on edge-based models \cite{Decreusefond2012,Janson2014}.  

This paper is organized as follows.  Section \ref{sec:models} introduces the stochastic model that is considered along with the necessary notation.  Section \ref{sec:lln} presents our main result, a law of large numbers for the stochastic process on the dynamic multilayer network, and considers several important special cases, which relate our result to edge-based and pairwise models.  The two-layer community-healthcare network model and its analysis is given in Section \ref{sec:CH}.  We conclude with a discussion in Section \ref{sec:discuss}.    The proof of our main theorem is given in Appendix \ref{sec:proof}  which also  provides further mathematical details \blue{along with a summary of notation} for the main body of the paper.
 
\section{Stochastic model} 
	\label{sec:models}

\blue{Recent advances in computational methods and the ever-increasing power of modern computers have made it possible to consider stochastic versions of the classical ODE epidemic models.  Such models not only provide the overall trend of an epidemic across a population but also inform about the stochastic fluctuations around the mean and, hence, about the intrinsic noise in the system (see, for instance  \cite{wilkinson2009stochastic} and references therein).  Some visual examples are provided  in the next section.  The stochastic  models  are  typically formulated as continuous time Markov processes with discrete state space (see, e.g., \cite{Altmann1995,Volz2008, Janson2014}), which is also  the framework we adopt in this paper. }

We start by introducing some notation and relevant definitions for dynamic multilayer networks.  
\red{The class of random graphs considered here will be an extension of the configuration model to the multilayer setting.  In the following section we will precisely define the {\it layered configuration model}, and extend the notions of degree and excess degree to the multilayer setting.}  Then, we introduce the stochastic process considered in this work, which is the appropriately modified version of the SIR process.  

\subsection{Layered configuration model} \label{sec:lcm} 
Let $r$ denote the number of layers and, for any vectors $\bx = (x_1,\hdots,x_r),\bk = (k_1,\hdots,k_r)$ in $\mathbb{R}^r$, denote  $\bx^\bk=\prod_{i=1}^r x_i^{k_i}$.  The probability generating function (pgf) of the multivariate degree distribution is given by  
\begin{equation}
\psi(\bx)=\sum_{\bk} p_\bk \bx^\bk \label{eq:psi}
\end{equation}
 where $p_\bk=P(k_1,\ldots, k_r)$ is the probability of  a node being of (multi-)degree $\bk$, i.e. having $k_i$ neighbors in layer $i$.

Given a realization of the degree distribution on $n$ nodes, we construct a multilayer graph as follows.  Each node is assigned a collection of half-edges in each layer corresponding to its degree, and then half-edges within each layer are paired uniformly at random.  \blue{We assume that the pairing   is  done independently in each layer.  Thus, restricted to the $j$-th layer, the  resulting graph is a realization of a configuration model (see, Chapter 13 in \cite{newman2010})  with the degree distribution given by the $j$-th marginal of $\psi$ (also see, for instance, Section 2.2.4 in \cite{Miller2013}).}  We refer to the collection of such realized graphs as the {\it layered configuration model} (LCM) and denote it by $\mathcal{G}_r(\psi,n)$.
 
\paragraph*{\em Excess degree distribution.} \hspace{0.01in}  In a single-layer graph, the \textit{excess degree} of a node $u$ is calculated by following an 
edge  to $u$ from a neighbor $v$ and counting the number of other neighbors  (not including $v$) of $u$ (see, Chapter 13 in \cite{newman2010}).  It will be 
convenient to extend the notion of an excess degree distribution to the multilayer setting.  Let $P_{j|i}(l)$ denote the probability that a randomly selected  $i
$-neighbor (i.e. neighbor in layer $i$) of a node $u$ has $j$-degree (i.e. degree in layer $j$) equal to $l$.  Then, by LCM construction, $P_{j|i}(l)$ is given as 
\[P_{j|i}(l) = \sum_{\bk: k_j = l} k_ip_\bk/\mu_i\]
where ${\mu_i} = \p_i\psi(\one) = \sum_\bk k_i p_\bk$
is the average $i$-degree, $\p_i$ denotes the partial derivative with respect to $x_i$, and $\one$ is the vector of ones in $\bR^r$.  Correspondingly, let $\psi^{ex}_{j|i}$ denote the pgf of the excess $j$-degree distribution of a node randomly selected as a $i$-neighbor.  Then, 
\[\psi_{j|i}^{ex}(x_j) = \sum_l P_{j|i}(l)x_j^l = \sum_{\bk} \dfrac{k_ip_\bk}{\mu_i}x_j^{k_j} = \dfrac{1}{\mu_i}\p_i\psi\blue{(1,...,1,x_j,1,...,1)} \]
where \blue{$(1,...,1,x_j,1,...,1)$} is the vector of ones with the $j$th coordinate replaced by $x_j$.   
The {\em average excess $j$-degree of an $i$-neighbor} is then given by 
\[\mu^{ex}_{j|i} = \p_j\psi^{ex}_{j|i}(1) = \dfrac{1}{\mu_i} \p^2_{ij}\psi(\one).\]
Finally, we define the  {\em normalized average excess $j$-degree of an $i$-neighbor}  as 
 \begin{equation}
  \kappa_{ji} =\dfrac{\mu^{ex}_{j|i}}{\mu_j} = \dfrac{\p^2_{ij}\psi(\one)}{\p_i\psi(\one)\p_j\psi(\one)}. \label{eq:kappa}
  \end{equation}
 Note that, for the univariate (i.e. single layer) case when $r=1$, $P_{i|i}(k)= kp_k/\mu$, which is the well-known distribution of the degree of a neighbor (also referred to as the size-biased distribution \cite{VanDerHofstad2009} and corresponding to the excess degree distribution $q_k = (k+1)p_{k+1}/\mu$~\cite{newman2010}), and $\kappa$ is the ratio of the average excess degree to average degree.

\subsection{$SIdaR$ process} \label{sec:SIdaR}

Assume that we have a realization of an LCM $\mathcal{G}_r(\psi,n)$ specifying the contact network for a population of size $n$.  The disease modeling framework adopted is the standard \blue{Markovian} SIR compartmental model where individuals are classified based on their infection status \cite{kermack1927}.  $S$, $I$ and $R$ correspond, respectively, to susceptible, infected, and recovered (or removed) individuals.  
We assume that edges within layers represent potentially infectious contacts of a certain type and we allow the network to be dynamic in response to infection.  
\blue{That is, we assume that infected nodes will either activate or deactivate their edges, depending on edge type}.  An infectious node drops (resp. activates) edges in layer $j$ at rate $\delta_j$ (resp. $\eta_j$).  We assume that a layer cannot be both activating and deactivating, i.e. at most one of $\delta_j$ and $\eta_j$ are nonzero, and we also assume that all deactivating layer edges are initially activated and all activating layer edges are initially deactivated.  \blue{Note that both active and deactivated edges are counted in a node's degree, which is therefore constant  throughout the course of an epidemic (see Section 2.3.1 in \cite{Miller2013} for a similar approach)}. Let $j=1,\hdots,m$ denote the deactivating layers (with $\eta_j = 0$) and let $j=m+1,\hdots,r$ denote the activating layers (with $\delta_j=0$).  Then, $2r+1$ event types may occur: infection (I) along an edge of any of the $r$ types, drop ($d$) of a deactivating edge or activation ($a$) of an activating edge, and recovery ($R$).  The timings of all events are assumed to follow independent exponential clocks with the following rates:
\begin{center}
\begin{tabular}{lll} 
 $\beta_j$ & rate of infection along $j$-\,edges ($S\stackrel{j}{\longrightarrow} I$),  & $j = 1,\hdots,r$\\
$\delta_j$ &  rate of deactivation (drop) of $j$-\,edges, & $j= 1,\hdots,m$\\
  $\eta_j$ & rate of activation of $j$-\,edges, & $j = m+1,\hdots,r$ \\
  $\gamma$ &rate of recovery ($I\longrightarrow R$). & \\ 
\end{tabular}
\end{center}

\begin{table}[tbh]
\caption{Transitions for the $SIdaR(r,m)$ process according to the $2r+1$ possible event types with corresponding rates.  Network arrangements corresponding to the transitions are also given with $\stackrel{j} -$ and $\stackrel{j} \sim$ denoting, respectively, active and deactivated edges of type $j$ between nodes (denoted $u$, $v$ and $w$). Here, $N_l^{Y,u}$ denotes the set of $l$-neighbors of node $u$ with disease status $Y$.}
\label{tab:sidar}
\begin{tabular}{llll}
\hline\noalign{\smallskip}
Event & Rate & Transition & Arrangement   \\
\noalign{\smallskip}\hline\noalign{\smallskip}
Infection of $u$ by $v$   & $\beta_j X^{SI}_j$ & $(X^S,X^I) \to (X^S - 1,X^I + 1)$ & \\
along $j$-edge  & & $(X_l^{SS},X_l^{SI}) \to (X_l^{SS} - X_l^{SS,u},X_l^{SI}+X_l^{SS,u})$  & $v \stackrel{j} -u \stackrel{l}- w$ \quad for  $w \in N_l^{S,u}$ \\
$j=1,\hdots,r$ & & $(X_l^{\SS},X_l^{\SI}) \to (X_l^{\SS} - X_l^{\SS,u},X_l^{\SI}+X_l^{\SS,u})$ & $v \stackrel{j} -u \stackrel{l} \sim w$\quad  for  $w \in N_l^{S,u}$ \\
& & $X_l^{SI} \to X_l^{SI} - X_l^{SI,u}$ & $v \stackrel{j} -u \stackrel{l}- w$\quad  for  $w \in N_l^{I,u}$ \\
& & $X_l^{\SI} \to X_l^{\SI}- X_l^{\SI,u}$ & $v \stackrel{j} -u \stackrel{l}\sim w$\quad  for  $w \in N_l^{I,u}$ \\
\noalign{\smallskip}\hline\noalign{\smallskip}
Deactivation of $j$-edge & $\delta_j X^{SI}_j$ & $X_j^{SI} \to X_j^{SI} -1$ & \\
 $j=1,\hdots,m$ &  & & \\
\noalign{\smallskip}\hline\noalign{\smallskip}
Activation of $j$-edge & $\eta_j X^{\SI}_{j}$ & $(X_j^{\SI},X_j^{SI}) \to (X_j^{\SI}-1,X_j^{SI}+1)$ & \\
$j=m+1,\hdots,r$ & &  & \\
\noalign{\smallskip}\hline\noalign{\smallskip}
Recovery of infected $u$ & $\gamma X^I$ & $X^I  \to X^I -1$  &\\
& & $X_j^{SI} \to X_j^{SI}  -X_j^{IS,u}$ & $u \stackrel{j} - w$ \quad for $w \in N_l^{S,u}$ \\
& & $X_j^{\SI} \to X_j^{\SI} -X_j^{\IS,u}$ & $u \stackrel{j} \sim w$ \quad for $w \in N_l^{S,u}$ \\
\noalign{\smallskip}\hline
\end{tabular}
\end{table}  

For a susceptible node $u$, let $X_j^{SI,u}$ and $X_j^{SS,u}$ denote, respectively, the number of infectious (i.e., infected,  not yet recovered) and susceptible active $j$-neighbors of $u$.  Similarly, let $X_j^{\SI,u}$ and $X_j^{\SS,u}$ denote the number of deactivated $j$-neighbors of $u$.  Also, for an infected node $u$, let $X_j^{IS,u}$ and $X_j^{\IS,u}$ denote the number of susceptible active and deactivated, respectively, $j$-neighbors of $u$.  We consider aggregate variables that are the total number of nodes or pairs of neighboring nodes (i.e. dyads) with a given disease status.  For example, the total number of $j$-edges between susceptible and infectious individuals is denoted $X_j^{SI}$ and is given by $X_j^{SI} = \sum_{u \in S} X_j^{SI,u}$.  We regard  the aggregate dyad counts as vectors in ${\bR}^r$, e.g. $\XSI = (X_1^{SI},\hdots,X_r^{SI})$ and likewise for $\XSIT$, $\XSS$, and $\XSST$.  Note that $\XSS$ and $\XSST$ count the edges twice.  We let $\X(t)= (X^S,X^I,\XSI,\XSIT,\XSS,\XSST)(t)$ denote the state of the aggregate stochastic process at time $t>0$ where $X^S$ and $X^I$ denote the number of susceptible and infectious nodes, respectively.  Note that the number of recovered individuals is given by $X^R = n-X^S-X^I$ and so, for the sake of simplicity, we ignore the equation for $X^R$ throughout.  \blue{In addition, we do not keep track of $\bm{X^{II}}$ or the dyads of recovered individuals since the evolutions of the main quantities of interest, $S$ and $I$, are not affected  by these variables.} The transitions for the aggregate process are listed in Table \ref{tab:sidar}.  We refer to such a process as $SIdaR(r,k)$ in order to emphasize the activation and deactivation events.  The analysis of this process is complicated, partially due to the aggregation of the nodes that destroys the Markov property (see, e.g. \cite{Altmann1998}).

Note that  the  dyad  (e.g. $\XSI$) is understood throughout as a (row) vector in $\mathbb{R}^r$.  Also, with a slight abuse of notation we take multiplication, division, integration and ordering of vectors to be coordinatewise.  The state variables depend on $n$ but we do not explicitly acknowledge this in our notation.  

Consider  the $SIdaR(r,k)$ process $\X(t)$ on the LCM $\G_r(\psi,n)$ with transitions as outlined in Table \ref{tab:sidar}. The Doob-Meyer decomposition theorem \cite{Meyer1962} guarantees the existence of a zero-mean martingale $\bm{M}(t)=(M^S,M^I,\bm{M^{SI}},\bm{M^{\SI}},\bm{M^{SS}},\bm{M^{\SS}})(t)$ such that  
\begin{equation}
\X(t) = \X(0) + \int_0^t \bm{\F^X}(\X(s))ds + \bm{M}(t) \label{eq:X}
\end{equation}
where the integrable function $\bm{\F^X}(\X) = (\F^S,\F^I,\bm{\F^{SI}},\bm{\F^{\SI}},\bm{\F^{SS}},\bm{\F^{\SS}})(\X)$ is given by 
\begin{align}\label{Feqn}
\F^S(\XSI)= &- \sum_{l=1}^r\beta_l X_l^{SI}, & \notag\\
\F^I(X^I,\XSI)= & \sum_{l=1}^r\left(\beta_l X_l^{SI}\right)-\gamma X^I, &\notag\\ \F_j^{SI}(\XSI,\XSIT,\XSS)= & \sum_{i \in S}\left(\sum_{l=1}^r \beta_lX_l^{SI,i} (X_j^{SS,i}-X_j^{SI,i})\right)-(\gamma+\delta_j)X_j^{SI} + \eta_jX^{\SI}_j, \notag\\
\F_j^{\SI}(\XSI,\XSIT,\XSST)= & \sum_{i \in S}\left(\sum_{l=1}^r \beta_lX_l^{SI,i} (X_j^{\SS,i}-X_j^{\SI,i})\right)-(\gamma+\eta_j)X_j^{\SI} + \delta_jX_j^{SI},\notag\\
\F_j^{SS}(\XSI,\XSS)= & -2\sum_{i\in S} \sum_{l=1}^r \beta_l X_l^{SI,i}X_j^{SS,i}, \notag\\
\F_j^{\SS}(\XSI,\XSST)= & -2\sum_{i\in S} \sum_{l=1}^r \beta_l X_l^{SI,i}X_j^{\SS,i}, 
\end{align}
for $j=1,\ldots,r$.

We now define two more variables that will help us describe the evolution of the process in the large graph limit.  Let $\XSdot (t) = (X_1^{S \bigcdot},\hdots,X_r^{S \bigcdot})(t)$ where $X_j^{S \bigcdot}(t)$ is the number of $j$-edges \blue{(active and deactivated)} belonging to susceptible nodes at time $t$.
We partition the collection of susceptible nodes $S$ by their degree $\bk\ge 0$ so that $S=\cup S_\bk$, which corresponds also to $X^S=\sum_\bk X^{S_\bk}.$ Note 
\begin{equation}
\XSdot (t)=\sum_\bk \bk X^{S_\bk}(t). \label{eq:XSdot}
\end{equation}

We also define
$\Xth(t) = (\xth_1,\ldots,\xth_r)(t)$ by
\begin{equation}\label{eq:th}
 \Xth(t)=\exp\left(-\bm{\beta}\int_0^t \dfrac{\XSI(s)}{\XSdot(s)} ds\right) 
\end{equation} 
where $\bm{\beta} = (\beta_1,\hdots,\beta_r)$.  We may interpret $\xth_j(t)$ as the probability of no infection along a $j$-edge by time $t$ in a susceptible node of $j$-degree one, given that the node was susceptible at $t=0$.  That is, $\Xth^{\one} = \prod_{j=1}^r \xth_j$ is the probability that a susceptible node of (multi-)degree $\one$ has not been infected through any layer by time $t$, given that the node was susceptible at $t=0$.  Note also that we may equivalently write 
\begin{equation}
\Xth(t) = \Xth(0) + \int_0^t \bm{\F}^{\Xth} (\XSI(s),\XSdot(s),\Xth(s))ds \label{eq:intth}
\end{equation}
where $\Xth(0) = \one$ and
\begin{equation}
\bm{\F}^{\Xth}(\XSI,\XSdot,\Xth) = -{\bm\beta}\Xth \XSI /\XSdot. \label{eq:Ftheta}
\end{equation}

As shown in Theorem \ref{thm:lln} below, $\Xth(t)$ plays a key role in describing the evolution of $\X(t)$ in the large graph limit.  The use of such a variable  was pioneered by Volz \cite{Volz2008} and Miller \cite{Miller2011b} in their edge-based approach.  In fact, as shown in Section \ref{sec:eb}, in the single-layer, static network case, the large graph limit of $\Xth(t)$ corresponds to the variable in the standard SIR edge-based model \cite{Miller2011}.

\section{Large graph limit theorems} 
	\label{sec:lln}

\blue{The stochastic process defined in Section \ref{sec:SIdaR} is complex and difficult to analyze directly.  In this section we present a limit theorem (Theorem \ref{thm:lln}) that shows, under mild technical assumptions, that the stochastic process converges to a relatively simple system of ODEs as the number of nodes tends to infinity.  The limiting ODEs retain key features of the epidemic process while being amenable to analysis.  In the case of a finite but large population, analysis of this deterministic system provides a good approximation to disease dynamics, in the sense that fluctuations around the mean due to intrinsic noise in the system shrink as the graph size grows. 
The study of such large volume limits for stochastic processes of the type  discussed here was originated  by Kurtz in \cite{kurtz1970solutions}  in the context of chemical reaction models.   His work has   subsequently  inspired multiple large volume results on the stochastic SIR-type models (see, e.g., Chapter 5 in  \cite{Andersson2000}). 
Among others,   Andersson (\cite{andersson1998limit}) derived limit theorems for a discrete-time random graph epidemic model under rather restrictive assumptions on the degree sequence of the random graph, such as finiteness of a $(4+\epsilon)$-th moment for some $\epsilon > 0$. Using a heuristic argument,  Volz (\cite{Volz2008}) presented scaling limits for an SIR model on random graphs in the form of ODEs. Decreusefond et al. (\cite{Decreusefond2012}) later proved Volz's results rigorously by summarizing the epidemic dynamics  on a configuration model random graph using  certain  measure-valued process. Several similar law of large numbers-type scaling limits under varying sets of technical assumptions surfaced afterwards. For example, Bohman and Picollelli (\cite{bohman2012sir})
 and Barbour and Reinert  (\cite{barbour2013approximating}) assumed uniformly bounded degrees. Janson, Luczak and Windridge \cite{Janson2014} assumed the degree of a randomly chosen susceptible vertex to be uniformly integrable and the maximum degree of initially infected nodes to be $o(n)$, a condition slightly less restrictive than our condition (A3)  below.  However,  none of these previous works have considered multiple layers in the random graph or allowed for activation/deactivation events.}

We start by formulating  our assumptions  in the general case when  evolution of the quantities of interest, $\X(t)$,  involves  a function of the variable $\Xth(t)$  defined in the preceding section.  In Section~\ref{sec:eb}, we state corollaries that relate our result to edge-based models in the special case of static graphs (that is,  in the absence of  activation or deactivation) \cite{Miller2013}.  Finally, in Section \ref{sec:pw} we show that, for a particular class of degree distributions, the evolution of  $\X(t)$ decouples from  $\Xth(t)$. \blue{This fact  reveals  an interesting  connection between our limiting system and the one obtained via pair approximation.  While others have recently investigated the conditions for exactness of \textit{local} network moment closures \cite{pellis2015exact,sharkey2015exact}, we are able to obtain the condition for  exactness of a \textit{population-level} network moment closure. We confirm that,  as  suggested in  \cite{Sharkey2008,pellis2015exact},  such a condition depends on  the network structure (i.e. degree) heterogeneity.  Indeed, in Section~\ref{sec:pw}, we employ Theorem~\ref{thm:lln} to prove that the pair approximation approach may or may not give the correct large graph limiting ODEs for the class of LCM stochastic processes described here.  In particular, we provide a necessary and sufficient condition 
on the degree distribution for the two limiting systems to coincide.}
\subsection{General case} \label{sec:gen}

 \blue{All  limits considered  below,  unless otherwise noted, are with respect to the number of nodes $n\to \infty$. } We use $\stackrel{P} \to$ to denote convergence in probability in the product space of the right-continuous with finite left limits (\textit{c\`{a}dl\`{a}g)} stochastic processes and the  space of all random configurations drawn according to an LCM $\mathcal{G}_r(\psi,n)$. We  say that a sequence of random variables $Y_n\to \infty $ \textit{with high probability  (w.h.p.)} if  $P(Y_n>k)\to 1$ for any $k>0$.  Let $0 < T < \infty$.  We make the following assumptions:
\begin{enumerate}[label=(A\arabic*),leftmargin=*]
\item For $0<t\le T$, $\XSdot(t)\to \infty$ $w.h.p.$ \label{A1}
\item The fractions of initially susceptible, infected, and recovered nodes converge, respectively, to some $\alpha^S$, $\alpha^I$, $\alpha^R$ $\in [0,1]$, i.e.
\[X^S(0)/n \stackrel{P} \to \alpha^S, \quad X^I(0)/n \stackrel{P} \to \alpha^I, \quad X^R(0)/n \stackrel{P} \to \alpha^R.\] 
Furthermore, $\alpha^S>0$, $\alpha^I>0$, and the initially infected and recovered nodes are chosen randomly. \label{A3}
\item $\sum_\bk ||\bk||^2 p_\bk<\infty.$ \label{A2}
\end{enumerate}

\red{The assumption \ref{A1} implies that, for large graphs, some proportion of the individuals remain susceptible on $[0,T]$ and, hence, $\Xth$ is well-defined in \eqref{eq:th}.}  Furthermore, it implies that the average $j$-degree of a randomly chosen node, i.e. $\p_j\psi(\one)$, is positive since $\bm{0} < \liminf n^{-1}\XSdot \le \bm{\p\psi}(\one)$.  

Assumption \ref{A3} implies that the initial conditions for the dyads, scaled by $n$, also converge in probability, i.e. \begin{equation}
X_j^{SI}(0)/n \stackrel{P} \to \alpha_j^{SI}, \quad X_j^{\SI}(0)/n \stackrel{P} \to \alpha_j^{\SI}, \quad, X_j^{SS}(0)/n \stackrel{P} \to \alpha_j^{SS}, \quad X_j^{\SS}(0)/n \stackrel{P} \to \alpha_j^{\SS}, \quad j = 1,\hdots,r \label{A4}
\end{equation}
where, for the deactivating layers $j = 1,\hdots,m$, 
\[\begin{array}{lllll}
\alpha_j^{SI} = \alpha^S\alpha^I \mu_j, \quad & \alpha_j^{\SI} = 0, \quad&  \alpha_j^{SS} = (\alpha^S)^2 \mu_j, \quad& \alpha_j^{\SS} = 0, \quad& \qquad j = 1,\hdots,m
\end{array}\]
and, for the activating layers $j = m+1,\hdots,r$,
\[\begin{array}{lllll}
\alpha_j^{SI} = 0, \quad& \alpha_j^{\SI} = \alpha^S\alpha^I \mu_j, \quad&  \alpha_j^{SS} = 0, \quad& \alpha_j^{\SS} = (\alpha^S)^2 \mu_j, \quad& \qquad j = m+1,\hdots,r.
\end{array}\]
To illustrate the above  consider,  for example,  the initial condition $\alpha^{SI}_j$ for $j=1,\hdots,m$ as follows. By assumption, all deactivating layer edges are initially activated and all activating layer edges are initially deactivated.  Then,  according to \ref{A3}, the limiting probability of selecting a random node that is susceptible is $\alpha^S$.  The average number of $j$-edges a node has is $\mu_j$, and the limiting probability that a given edge connects to an infected node is $\alpha^I$.  Therefore, $\alpha^{SI}_j = \alpha^S\alpha^I\mu_j$.  The other dyad initial conditions are obtained similarly. We denote $\bm{\alpha} = (\alpha^S,\alpha^I,\bm{\alpha^{SI}},\bm{\alpha^{\SI}},\bm{\alpha^{SS}},\bm{\alpha^{\SS}})$ in what follows.  

 \blue{The assumption \ref{A2} implies that $\sum_\bk k_j^2p_\bk < \infty$ for $j=1,\hdots,r$ (i.e. the second moments of the marginal degree distributions are finite) and consequently also  that, in each layer, the multigraph constructed by matching half-edges uniformly at random is a simple graph with positive probability (see Section 2 in \cite{Janson2014}). Therefore \ref{A2} also guarantees  that there is a positive probability of generating a simple LCM graph, i.e. a multilayer graph with a simple graph in each layer.}
 
Before stating the main result, we define a quantity that plays a key role in describing how the network structure affects the large graph limit.  For $ 1 \le j,l \le r$, let $\bar{\kappa}_{jl}$ be defined by
\begin{equation} \label{eq:kappabar}
\bar{\kappa}_{jl}(\bx) = \frac{\psi(\bx)\p^2_{jl}\psi(\bx)}{\p_j\psi(\bx)\p_l\psi(\bx)}.
\end{equation}
As we discuss in Section \ref{sec:correct}, $\bar{\kappa}_{jl}(\Xth)$ can be interpreted as the ratio of the average excess $j$-degree of a susceptible node chosen randomly as an $l$-neighbor of an infectious node to the average $j$-degree of a susceptible node. 

We now  define the function that in  Theorem \ref{thm:lln} below  describes the evolution of $(\X,\Xth)(t)$ in the large graph limit.  Let $(\bx,\th)=(S,I,\bm{[SI]},\bm{[\SI]},\bm{[SS]},\bm{[{\SS}]},\bm{\th})$ and define $\bm{\H}(\bx,\th)=(\bm{\H}^{\X},\bm{\H}^{\Th})(\bx,\th)$ where $\bm{\H}^{\X} = (\H^S,\H^I,\bm{\H^{SI}},\bm{\H^{\SI}},\bm{\H^{SS}},\bm{\H^{\SS}})$ and $\bm{\H}^{\Th}$  are given by
\begin{equation} \label{eq:H}
\begin{aligned}
\H^S(\bm{[SI]})= &- \sum_{j=1}^r\beta_j [SI]_j, & \\
\H^I(I,\bm{[SI]})= & \sum_{j=1}^r\left(\beta_r [SI]_r\right)-\gamma I, &\\
\H_j^{SI}(\bm{[SI]},\bm{[\SI]},\bm{[SS]},\th)= & \sum_{l=1}^r\left[ \beta_l\bar{\kappa}_{jl}(\th)\frac{[SI]_l}{S} ([SS]_j-[SI]_j)\right]-(\beta_j+\gamma+\delta_j)[SI]_j + \eta_j[\SI]_j, \\ 
\H_j^{\SI}(\bm{[SI]},\bm{[\SI]},\bm{[\SS]},\th)= & \sum_{l=1}^r\left[ \beta_l\bar{\kappa}_{jl}(\th)\frac{[SI]_l}{S} ([\SS]_j-[\SI]_j)\right]-(\gamma+\eta_j)[\SI]_j + \delta_j[SI]_j, \\
\H_j^{SS}(\bm{[SI]},\bm{[SS]},\th)= & -2\sum_{l=1}^r \beta_l \bar{\kappa}_{jl}(\th) \frac{[SI]_l[SS]_j}{S},\\ 
\H_j^{\SS}(\bm{[SI]},\bm{[\SS]},\th)= & -2\sum_{l=1}^r \beta_l \bar{\kappa}_{jl}(\th) \frac{[SI]_l[\SS]_j}{S}, \\
\H_j^\theta(\bm{[SI]},\th)=& -\beta_j \frac{[SI]_j}{\alpha^S\p_j\psi(\th)},
\end{aligned}
\end{equation}
for $j=1,\hdots,r$.

\begin{theorem}[Strong law of large numbers] \label{thm:lln}
Assume conditions \ref{A1}$-$\ref{A2} for the LCM $\G_r(\psi,n)$.  Then, for any $0 < T < \infty$,
\[\sup_{0<t\le T} \norml (\X/n,\Xth)(t)-(\bx,\th)(t) \normr \stackrel{P}{\rightarrow} 0\]
where  $(\bx,\th)(t)$ is the solution of  
\begin{equation}
(\bx,\th)(t)=(\bx,\th)(0)+\int_0^t\H((\bx,\th)(s))ds \label{eq:D}
\end{equation}
with initial conditions $\bx(0)=\bm{\alpha}$ and $\th(0)=\one.$ 
\end{theorem}

Theorem \ref{thm:lln} specifies the large graph limit of the aggregated $SIdaR(r,m)$ process on $\G_r(\psi,n)$ under conditions  \ref{A1}$-$\ref{A2}. It says that  $(\X/n,\Xth(t))$ converges uniformly in probability on any finite interval $[0,T]$ to the solution $(\bx,\th)(t)$ of the deterministic set of equations given by \eqref{eq:D}.  \blue{Proof of Theorem~\ref{thm:lln} is provided in Appendix~\ref{sec:proof}. The argument   largely follows the standard large volume  analysis \cite{kurtz1970solutions}. The main difficulty is in  showing that  the  terms corresponding to ``empirical moments"  in system \eqref{Feqn} (i.e. those summing over susceptible nodes in the graph and counting exact numbers of neighbors of certain type)  can be replaced  in the large graph limit \eqref{eq:H} with the corresponding $\bar{\kappa}$ terms  encapsulating a population-level average over the heterogeneity in the degrees and excess degrees of susceptible nodes. This may be done due to  the properties of the LCM construction with the help of the representation introduced in Remark~\ref{rem:hg} in Appendix~\ref{sec:proof}.}

\subsection{Edge-based limiting systems} \label{sec:eb}

We consider two special cases of the large graph limit theorem for multilayer networks.  First, we consider a static network, i.e. the case where $\delta_j = \eta_j =0$ for $j=1,\hdots,n$.  Corollary \ref{cor:EBequiv} below states that in this case our system \eqref{eq:D} is equivalent to an edge-based model with multiple modes of transmission.  This model is the one  proposed by Miller and Volz \cite{Miller2013} but with a modification to allow for a large number of initially infected nodes (following \cite{Miller2014b} where Miller modifies the standard SIR edge-based model for such a scenario).  In the case that the initially infected nodes are randomly chosen (which we assume in \ref{A3}), the model is given by
\begin{equation} \label{eq:EB}
\begin{array}{ll}
\dfrac{d\theta_j}{dt} = -\beta_j\theta_j + \beta_j\alpha^S\dfrac{\p_j\psi(\th)}{\p_j\psi(\one)} + \gamma(1-\theta_j) + \beta_j\alpha^R, \qquad j = 1,\hdots,r  \\[0.1in]
\dfrac{dR}{dt} = \gamma I, \qquad S = \alpha^S\psi(\th), \qquad I = 1-S-R, \\[0.1in]
\th(0) = \one, \qquad R(0) = \alpha^R, \qquad S(0) = \alpha^S, \qquad I(0) = 1-\alpha^S-\alpha^R.
\end{array}
\end{equation}
The   statement of Corollary \ref{cor:EBequiv} is as follows. Its proof is  given in Appendix \ref{app:equiv}.
\begin{corollary} \label{cor:EBequiv}
Assume $\delta_j = \eta_j =0$ for $j=1,\hdots,n$ and the conditions of Theorem \ref{thm:lln} hold.  Then, the conclusions of Theorem \ref{thm:lln} hold where $(\bx,\th)(t)$ is equivalent to the solution of the edge-based model with multiple modes of transmission \eqref{eq:EB}.
\end{corollary}

We further consider the special case of a static, single layer graph (i.e. $r=1$ and $\delta_1 = \eta_1 = 0$).  In the case $r=1$, \eqref{eq:EB} reduces to the well-known edge-based SIR model of Volz and Miller et al. \cite{Miller2011,Volz2008}, which has been proven to be the large graph limit of the SIR stochastic process on a static configuration model graph \cite{Janson2014,Decreusefond2012}.
By taking $r=1$ in Corollary \ref{cor:EBequiv} we have provided an alternative proof of this fact.

\begin{corollary} \label{cor:EBequiv2}
Assume $r=1$, $\delta_1 = \eta_1=0$ and the conditions of Theorem \ref{thm:lln} hold.  Then, the conclusions of Theorem \ref{thm:lln} hold where $(\bx,\theta)(t)$ is equivalent to the solution of the edge-based SIR model \eqref{eq:EB} with $r=1$.
\end{corollary}

\subsection{Pairwise limiting systems}  \label{sec:pw}

In this section we consider a certain class of LCMs where $\bar{\kappa}_{ij}$ as defined by equation \eqref{eq:kappabar} is constant.  This affords a substantial simplification to the limiting system \eqref{eq:H} and, in fact, the system of differential equations defining the large graph limit will coincide with the model derived via the pairwise approach. 

\subsubsection{Poisson-type distributions}

 We define a multivariate {\em Poisson-type (PT) distribution} to be a distribution with a pgf that satisfies 
\begin{equation} \label{eq:mpt}
\bar{\kappa}_{ij}(\bx) \equiv \bar{\kappa}_{ij}(\one) = \kappa_{ij}, \qquad i,j = 1,\hdots,r
\end{equation}
where we recall the definition of the normalized average excess degree, $\kappa_{ij}$, in equation \eqref{eq:kappa}.
 At first glance, \eqref{eq:mpt} may seem opaque; however, in fact it is satisfied by a broad class of distributions.  For example, in the single layer case this condition is equivalent to 
\[\psi'(x) = \psi'(1)\psi^\kappa,\]
which is satisfied by the univariate Poisson ($\kappa=1$), binomial ($\kappa<1$), and negative binomial ($\kappa>1$) distributions.  Note that a $k$-regular graph (where all nodes have degree $k$) may be considered  a special case of the binomial distribution and that  the geometric distribution is a special case of the negative binomial distribution.  Bansal et al.  \cite{Bansal2007}  have shown that the geometric degree distribution (i.e. the discrete analog of the exponential distribution) gives the best fit for several empirical contact networks.  In the multilayer case with $r>1$, \blue{if the marginal degree distributions in different layers  are independent}, i.e. the pgf $\psi(\bx)$ can be written as $\psi(\bx) = \prod_{j=1}^r\psi_r(x_r)$, the PT condition \eqref{eq:mpt} \blue{reduces to the degree distribution in each layer being of  (univariate) Poisson-type.}

\subsubsection{Pairwise model} \label{subsec:pw}

If the degree distribution is PT as defined by condition \eqref{eq:mpt}, the limiting system  \eqref{eq:H} defining $\bm{\H}=(\bm{\H}^{\X},\bm{\H}^{\Th})$ has a particularly simple form.   Indeed, substituting the constant $\kappa_{jl}$ for $\bar{\kappa}_{jl}(\th)$ decouples $\bm{\H}^\X$ from $\bm{\H}^\Th$.  We consider the resulting model in this section and introduce some new notation to do so.  Let $[XY]_j$ and $[\widetilde{XY}]_j$ denote, respectively, the number of activated and deactivated edges of type $j$ between a node of status $X$ (either $S$, $I$ or $R$) and a node of status $Y$.  Let $[XYZ]_{ij}$ denote the number of triples with an active $i$-edge between nodes of status $X$ and $Y$ and an active $j$-edge between the node of status $Y$ and one of status $Z$.  Similarly, $[\widetilde{XY} Z]_{ij}$ will denote such triples where the $i$ edge is deactivated.

Under the correlation equations approach of Rand \cite{Rand1999}, triples are needed to describe the evolution of pairs, quadruples (e.g. $[XYZW]_{ijk}$) are needed to describe the evolution of triples, and so forth.  A pair approximation for triples is used in order to close the system at the level of pairs \cite{Rand1999}.  For consideration of triples in the multilayer setting, we must take into account the edge types and the appropriate excess degrees as defined in Section \ref{sec:lcm}. Let $p(u=X|A)$ denote the probability that a node $u$ has disease status $X \in \{S,I,R\}$ given an edge (or triple) arrangement $A$ for $u$.  The pair approximation of $[XYZ]_{ij}$ is then calculated as follows:
\begin{align}
[XYZ]_{ij} &= [u YZ]_{ij}p(u=X|[u YZ]_{ij}) \nonumber \\ &  \approx [uYZ]_{ij}p(u=X|[uY]_i) = \mu^{ex}_{i|j}[YZ]_{j}\dfrac{[XY]_i}{\mu_iY} = \kappa_{ij}\dfrac{[XY]_i[YZ]_j}{Y} \label{eq:pa}
\end{align}
with $\kappa_{ij}$ as defined in equation \eqref{eq:kappa}.  \blue{Note that since alternatively  we could have  approximated  $[XYZ]_{ij} \approx [XYu]_{ij}p(u=Z|[Yu]_j)$ we must have  that $\kappa_{ji} = \kappa_{ij}$. }

Applying the correlation equations approach using the pair approximation \eqref{eq:pa} to the $SIdaR(r,m)$ dynamics described in Section \ref{sec:SIdaR} results exactly in the same equations as the limiting system \eqref{eq:D} in the case of a PT distribution:
\begin{equation} \label{eq:pw}
\begin{aligned}
\dfrac{dS}{dt} &= -\sum_{j=1}^r \beta_j[SI]_j, & \\
\dfrac{dI}{dt} &=\sum_{j=1}^r \beta_j[SI]_j -\gamma I, &\\
\dfrac{d[SI]_j}{dt} &= \sum_{l=1}^r \beta_l\kappa_{jl}\dfrac{[SS]_j[SI]_l}{S} - \sum_{l=1}^r \beta_l\kappa_{jl}\dfrac{[SI]_j[SI]_l}{S}-(\beta_j+\gamma+\delta_j)[SI]_j + \eta_j[\SI]_j, & j = 1,\hdots,r, \\
\dfrac{d[\SI]_j}{dt} &= \sum_{l=1}^r \beta_l\kappa_{jl}\dfrac{[\SS]_j[SI]_l}{S} -\sum_{l=1}^r \beta_l\kappa_{jl}\dfrac{[\SI]_j[SI]_l}{S}-(\eta_j+\gamma)[\SI]_j+ \delta_j[SI]_j, & j = 1,\hdots,r,\\
\dfrac{d[SS]_j}{dt} &= -2\sum_{l=1}^r \beta_l\kappa_{jl}\dfrac{[SS]_j[SI]_l}{S} , & j = 1,\hdots,r, \\
\dfrac{d[\SS]_j}{dt} &= -2\sum_{l=1}^r \beta_l\kappa_{jl}\dfrac{[\SS]_j[SI]_l}{S},  & j = 1,\hdots,r.
\end{aligned}
\end{equation}
Here, we have derived the system \eqref{eq:pw} using absolute pair and triple counts.  Notice that, if we scale all variables by the graph size $n$, the nondimensional quantities satisfy the same system of equations (this holds for any $n$ and, hence by continuity of the solution, also in the limit).  From here on we will consider only the scaled variables, which will be convenient when we state the law of large numbers in Corollary~\ref{cor:PT}.  Accordingly, we set the initial conditions to be 
\begin{equation} \label{eq:pwic}
(S,I,[SI]_1,\hdots,[SI]_r,[\SI]_1,\hdots,[\SI]_r,[SS]_1,\hdots,[SS]_r,[\SS]_1,\hdots,[\SS]_r)(0) = \bm{\alpha}
\end{equation}
so that they agree with the initial conditions in Theorem \ref{thm:lln} as $n \to \infty$.

Observe also that, in fact, we can reduce the dimension of the system \eqref{eq:pw} since we only need to keep track of the deactivated edges for the activating layers, i.e. $j=m+1,\hdots,r$.  Also, $[SS]_j \equiv 0$ for an activating layer since its initial condition is zero and, hence, we only must track $[SS]_j$ for $j=1,\hdots,m$.  We refer to system \eqref{eq:pw} with its initial condition \eqref{eq:pwic} as the {\it pairwise model}.
 
The discussion above is summarized in the following corollary.
\begin{corollary} \label{cor:PT}
Assume the conditions of Theorem \ref{thm:lln} hold for LCM $\G_r(\psi,n)$.  Then, the conclusions of Theorem \ref{thm:lln} hold where $(\bx,\th)(t)$ is the solution of the pairwise model \eqref{eq:pw} if and only if $\G_r(\psi,n)$ has a multivariate Poisson-type degree distribution.
\end{corollary}
Furthermore, we can consider the implications of Corollary \ref{cor:PT} in the static, single layer case.  If $r=1$ and $\delta_1=\eta_1=0$, then the pairwise model \eqref{eq:pw} reduces to the well-known correlation equations model of Rand \cite{Rand1999}:
\begin{equation} \label{Rand}
\begin{aligned}
\dfrac{dS}{dt} &= -\beta [SI], \\
\dfrac{dI}{dt} &= \beta[SI] - \gamma I, \\
\dfrac{d[SI]}{dt} &= \beta\kappa\dfrac{[SS][SI]}{S}-\beta\kappa\dfrac{[SI]^2}{S}-(\beta+\gamma)[SI],\\
\dfrac{d[SS]}{dt} &= -2\beta\kappa\dfrac{[SS][SI]}{S}. 
\end{aligned}
\end{equation}

\begin{corollary} \label{cor:PT2}
Assume the conditions of Theorem \ref{thm:lln} hold for $\G_1(\psi,n)$, a static graph (i.e. $\delta_1=\eta_1=0$).  Then, the conclusions of Theorem \ref{thm:lln} hold where $(\bx,\theta)(t)$ is the pairwise SIR model of Rand \eqref{Rand} if and only if $\G_1(\psi,n)$ has a univariate Poisson-type degree distribution.
\end{corollary}

Note that together Corollaries \ref{cor:EBequiv} and \ref{cor:PT} imply that, in the case of a multivariate PT degree distribution on a static graph, the pairwise model \eqref{eq:pw} is equivalent to the edge-based model with multiple modes of transmission \eqref{eq:EB}.  Likewise, Corollaries \ref{cor:EBequiv2} and \ref{cor:PT2} indicate that the pairwise SIR model \eqref{Rand} is equivalent to the edge-based SIR model, \eqref{eq:EB} with $r=1$, when the distribution is PT.  We note that the edge-based SIR model has previously been shown to be equivalent \cite{House2011,Miller2014} to a higher dimensional pairwise model of Eames and Keeling \cite{Eames2002}.  The latter model stratifies the susceptible nodes by degree and, hence, has dimension $K+3$ where $K$ represents the number of distinct degrees \cite{Miller2014}. The model of dimension $K+3$ was derived as an approximation to an earlier well-known model of Eames and Keeling, which is of dimension $3K^2/2 + 3K/2+1$ \cite{Eames2002}.  We observe that, in fact, \eqref{Rand} can be reduced to two differential equations.  Separation of variables on $d[SS]/dS$  (see, e.g., \cite{Andersson2000} for $\kappa=1$) gives $[SS]=(\alpha^S)^{2(1-\kappa)}\mu S^{2\kappa}$.  Subsequent inspection of $d[SI]/dS$ yields a linear differential equation that can be solved to express $[SI]$ explicitly as a function of $S$: 
\begin{equation} \label{eq:SI}
\begin{aligned}[]
[SI] &= \dfrac{\beta+\gamma}{\beta(1-\kappa)}S - \mu S^{2\kappa} - \left( \dfrac{\beta+\gamma}{\beta(1-\kappa)}(\alpha^S)^{1-\kappa} - \mu (\alpha^S)^{\kappa} - (\alpha^S)^{1-\kappa}\alpha^I\mu\right)S^\kappa, \qquad \kappa \neq 1, \\
[SI] &= \dfrac{\beta+\gamma}{\beta}S\log(S) - \mu S^2 - \left( \dfrac{\beta+\gamma}{\beta}\log(\alpha^S) - \mu \alpha^S - \alpha^I\mu\right)S, \qquad \kappa = 1.
\end{aligned}
\end{equation}
  Thus, in Corollary \ref{cor:PT2}, we have identified a condition on the degree distribution under which the dimension of a pairwise model that is equivalent to the edge-based SIR model has been reduced from $K+3$ to two.

\subsubsection{Large-graph-consistent pair approximation} \label{sec:correct}
Corollary \ref{cor:PT} and the derivation of the pair approximation \eqref{eq:pa} motivate us to consider a more careful approximation of the triples in the general case when the distribution is not necessarily PT.  Let $\mu^{ex|SI}_{i|j}$ denote the average excess $i$-degree of a susceptible node chosen randomly as a $j$-neighbor of an infectious node and let $\mu^S_i$ denote the average $i$-degree of a susceptible node.  We now make more precise our comment in Section \ref{sec:gen} that $\bar{\kappa}_{ij}$ can be interpreted as the limiting ratio of these quantities.  

For the LCM $\G_r(\psi,n)$ with $\Xth$ as defined in equation \eqref{eq:th}, we derive $\mu^S_i$ and $\mu^{ex|SI}_{i|j}$ (see Appendix \ref{app:kappabar}) to be
\begin{equation}
\mu^S_i = \dfrac{\xth_i\p_i\psi(\Xth)}{\psi(\Xth)} \qquad \text{and} \qquad \mu^{ex|SI}_{i|j} = \dfrac{\xth_i\p^2_{ij}\psi(\Xth)}{\p_j\psi(\Xth)}. \label{eq:muS}
\end{equation}
Thus, we see from the definition of $\bar{\kappa}$ in \eqref{eq:kappabar} that
\[\bar{\kappa}_{ij}(\Xth) = \dfrac{\mu^{ex|SI}_{i|j}}{\mu^S_i}.\]
We note that $\mu^S_i$ and $\mu^{ex|SI}_{i|j}$ are dependent on time and differ, respectively, from $\mu_i$ and $\mu^{ex}_{i|j}$ defined in Section \ref{sec:lcm}. Indeed, susceptible nodes of high degree are preferentially infected \cite{May2001}.  On the other hand, the naive approximation of a triple with $\kappa_{ij} = \mu^{ex}_{i|j}/\mu_i$ in \eqref{eq:pa} uses the average degree and excess degree over all nodes, which remain constant for all $t \geq 0$.  Note that it is only necessary to approximate triples of the form $[XSI]_{ij}$ where $X \in \{S,I\}$.  Therefore, we can more carefully derive the pair approximation using $\mu^S_i$ and $\mu^{ex|SI}_{i|j}$
\begin{align}
& [XSI]_{ij} = [u SI]_{ij}p(u=X|[u SI]_{ij}) \nonumber\\  &\approx [uSI]_{ij}p(u=X|[uS]_i) = \mu^{ex|SI}_{i|j}[SI]_j\dfrac{[XS]_i}{\mu^S_i S} = \bar{\kappa}_{ij}(\Xth)\dfrac{[XS]_i[SI]_j}{S}. \label{eq:pairapproxS}
\end{align}
Note that the approximation $p(u=X|[u SI]_{ij})  \approx p(u=X|[uS]_i)$ becomes exact in the large graph limit as the configuration model becomes locally treelike \cite{Miller2014,sharkey2015exact}.

Moreover, it follows from Theorem \ref{thm:lln} that $\bar{\kappa}_{ij}(\Xth) \stackrel{P} \to \bar{\kappa}_{ij}(\th)$
uniformly on any finite interval $[0,T]$. We then see from  \eqref{eq:pairapproxS} that, if we took the correlation equations approach for $SIdaR(r,m)$ dynamics on a general LCM but instead approximated the triples with
\begin{equation}
[XYZ]_{ij} = \bar{\kappa}_{ij}(\th)\dfrac{[XY]_i[YZ]_j}{Y},  \label{eq:correct}
\end{equation}
the system of equations obtained would be exactly that given by the limiting system \eqref{eq:D}.  In this sense, the pair approximation \eqref{eq:correct} is ``correct", i.e. it is consistent with the large graph limit.  

\blue{Subsequently, we can view the pairwise system \eqref{eq:pw} as resulting from the approximation
\begin{equation}
\bar{\kappa}_{ij}(\th) \approx \kappa_{ij}.\label{eq:kappaapprox}
\end{equation}
For PT distributions, the above relation  is exact (``$\approx$'' may be replace with ``=''), and the pairwise system is consistent with the large graph limit (i.e. Corollary \ref{cor:PT}).  However, for non-PT distributions, the approximation \eqref{eq:kappaapprox} introduces some error (see 
Figure~\ref{fig:kappabar}).}

\begin{figure}
\includegraphics[scale = 0.38]{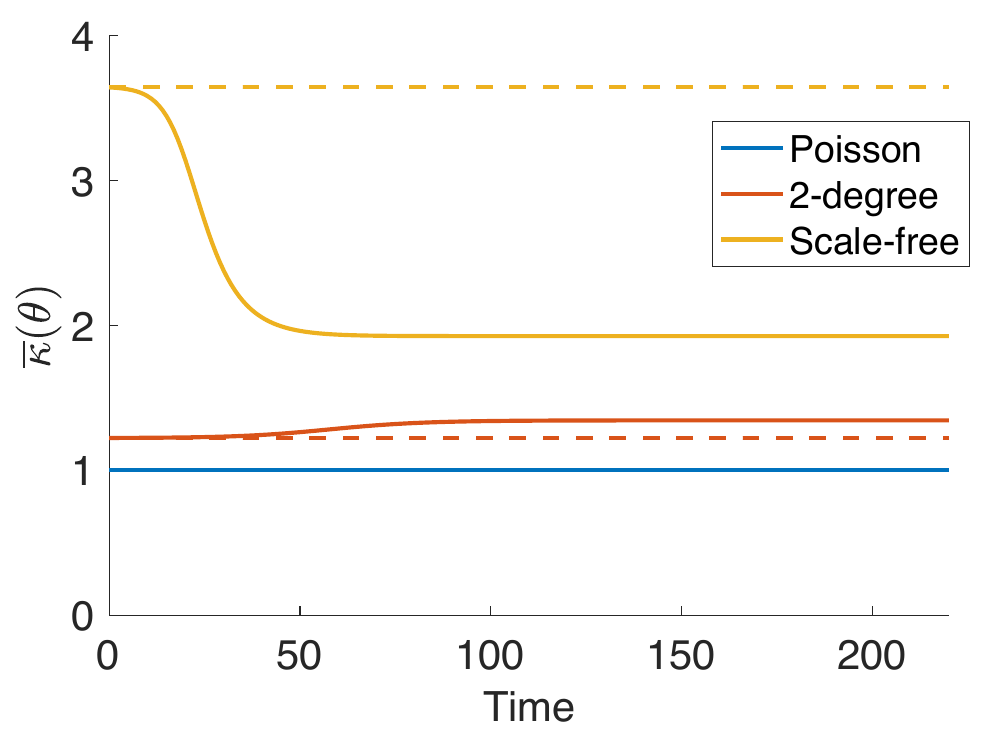}\quad 
\label{kappafig:b}
\includegraphics[scale=0.45]{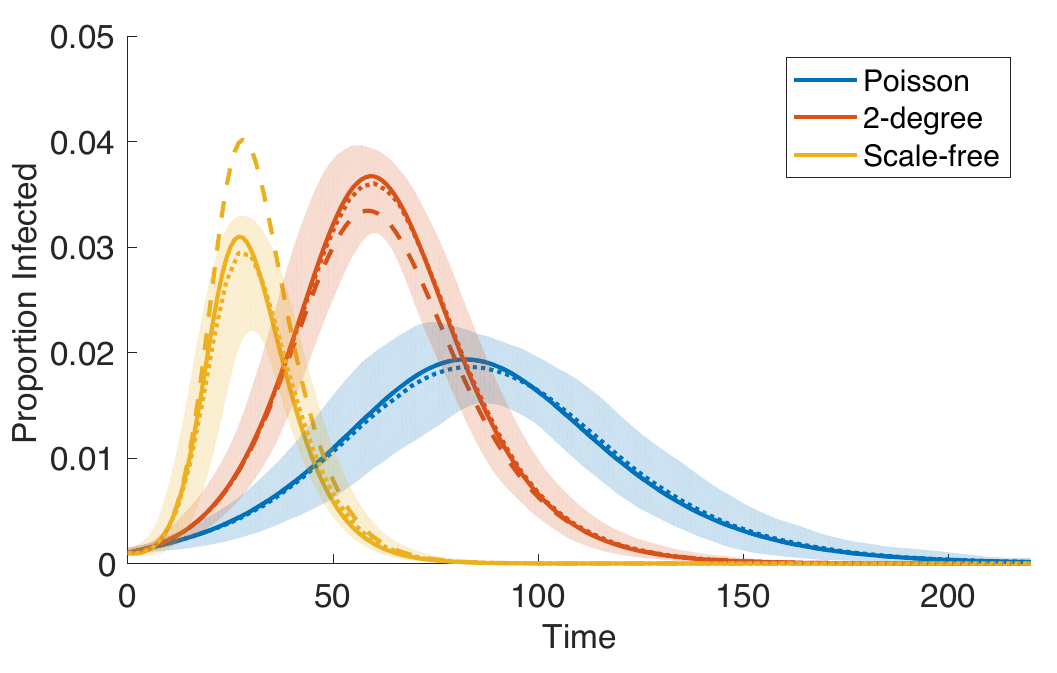}
\label{kappafig:c}

\caption{\blue{SIR epidemic simulations on static, single-layer graphs with degree distributions given by $\operatorname{Poisson}(6)$ (blue), a \textit{2-degree} distribution in which $p_2 = p_9 = 0.5$ (orange), and a scale-free distribution with $p_k = 1.6168k^{-2.01}$ for $1 \leq k \leq 50$ (yellow).  (Left) $\bar{\kappa}(\theta)$ (solid line) is constant and equal to $\kappa$ (dashed line) for the Poisson distribution but not for the non-PT distributions.  (Right) Solutions for the proportion of infected from the large graph limiting system (solid) and the pairwise system (dashed) show discrepancy for the non-PT distributions.  Approximate 95\% confidence intervals based on 100 stochastic simulations (shaded) and their means (dotted) show agreement with the large graph limiting systems.}}
\label{fig:kappabar}
\end{figure}

\blue{We illustrate the results of this section in Figure \ref{fig:kappabar} by simulating an SIR epidemic on three different static, single-layer graphs.  The first graph has a Poisson degree distribution with $\mu = 6$ (blue).  In the second graph, which we refer to as a \textit{2-degree} graph,  a node has either degree 2 or 9, each with probability 0.5 (orange).  The final graph is a so-called scale-free graph \cite{pastor-satorras2001}; that is, the degree distribution is given by a truncated power law distribution with $p_k = 1.6168k^{-2.01}$ for $1 \leq k \leq 50$ (yellow).  For each scenario, we simulate the stochastic SIR process on a graph of size $n=10^5$ and solve the edge-based model \eqref{eq:EB} (which, by Corollary \ref{cor:EBequiv2}, is equivalent to the large graph limiting system \eqref{eq:D}) as well as the pairwise system \eqref{eq:pw}.  In Figure \ref{kappafig:b}, we plot the quantity $\bar{\kappa}(\theta)$ (solid line) as a function of time with the starting value, $\kappa$, indicated with a dashed line.  The Poisson degree distribution is a PT distribution and, therefore, $\bar{\kappa}(\theta) \equiv \kappa$; however the latter two graphs do not have PT degree distributions and, indeed, we see a divergence of $\bar{\kappa}(\theta)$ from $\kappa$ as the outbreak progresses.  This corresponds to a discrepancy between the epidemic curves of the deterministic large graph limiting system (solid) and pairwise system (dashed) in Figure \ref{kappafig:c}.  Approximate 95\% confidence intervals  (shaded) based on 100 stochastic simulations are also shown in Figure \ref{kappafig:c}.  For the scale-free graph, $\kappa > \bar{\kappa}(\theta)$ after the initial time,  so using the ``naive" approximation \eqref{eq:pa} in  the pairwise system leads to overestimating  the number of triples.  Figure \ref{kappafig:c} indicates that this consequently results in an overestimate of the incidence of the epidemic (yellow, dashed).  In the 2-degree network, on the other hand, $\kappa$ slightly underestimates $\bar{\kappa}$ and the disease incidence is underestimated by the pairwise system (orange, dashed), although still within the margin of error indicated by the confidence intervals.  For the Poisson distribution, the two deterministic curves coincide. The means of the stochastic simulations  (dotted) show good agreement with the large graph limiting systems, following Theorem \ref{thm:lln}. }

\section{Community-healthcare network example} 
	\label{sec:CH}
	\blue{In this section, we discuss the two-layer dynamic network model that was mentioned in the Introduction as a concrete and  tractable example of the stochastic $SIdaR$ process.  We use this example to illustrate how our Theorem~\ref{thm:lln} can be applied to gain insight into disease dynamics, even in the quite complex setting where network dynamics are tied to infection status.  
One of the epidemiological issues we aim to understand is how the interplay between edge activation/deactivation and the multilayer network structure affects the ability of disease to invade.  In addition, gauging the sensitivity of outbreak size to network parameters and transmissibility along the different edge types would be useful for informing interventions. We demonstrate here that analysis of the pairwise limiting system provides answers to these questions, via application of well-established techniques for analyzing compartmental ODE models of disease dynamics.}

Consider the $SIdaR(2,1)$ process, as described in Section \ref{sec:SIdaR}, on a two-layer LCM $\G_2(\psi,n)$ with multivariate PT degree distribution where the two edge types correspond to community and healthcare type contacts.  We assume that infected individuals deactivate their community contacts, for example due to decreased mobility or isolation, while they activate their healthcare contacts as they seek care (Figure \ref{fig:nbhd}).   Note that the healthcare network may include both care provided by healthcare professionals at hospitals or other facilities as well as care provided in the home.  This model is motivated by the recent 2014$-$2015 outbreak of Ebola virus in West Africa.  The multitype contact features are particularly relevant for Ebola, given the disproportionate Ebola risk experienced by healthcare workers \cite{coltart2015,matanock2014} and women (primary caregivers in the home in West Africa) in the 2014 West Africa outbreak as well as Ebola outbreaks in the Democratic Republic of the Congo \cite{kratz2015}.  
\red{There are other aspects of Ebola that are important for transmission, such as an incubation period ranging from 2 to 21 days \cite{goeijenbier2014} and disease transmission at funerals \cite{Chowell2014,goeijenbier2014,Legrand2007}.  For simplicity we focus here on community and healthcare transmission and the effect of illness on network structure.}

Let $C$ denote community edges and $H$ denote healthcare edges.  As in Section \ref{sec:SIdaR}, the stochastic events are assumed to follow independent exponential clocks where the rates of infection along $C$- and $H$-edges are, respectively, $\beta_C$ and $\beta_H$, the rate of deactivation of a $C$-edge is $\delta$, the rate of activation of an $H$-edge is $\eta$ and the rate of recovery is $\gamma$.  

The $SIdaR(2,1)$ stochastic process above satisfies the conditions of Theorem~\ref{thm:lln} and Corollary~\ref{cor:PT} and, thus, converges to the following system of ODEs in the large graph limit:
\begin{equation}
\begin{aligned}
\dfrac{dS}{dt} &= -\beta_C[SI]_C-\beta_H[SI]_H  \\
\dfrac{dI}{dt} &= \beta_C[SI]_C+ \beta_H[SI]_H-\gamma I \\
\dfrac{d[SI]_C}{dt} &= \beta_C\kappa_{CC}\dfrac{[SS]_C[SI]_C}{S} + \beta_H\kappa_{CH}\dfrac{[SS]_C[SI]_H}{S} - \beta_{C}\kappa_{CC}\dfrac{[SI]^2_C}{S}-\beta_H\kappa_{CH}\dfrac{[SI]_C[SI]_H}{S}-(\beta_C+\gamma+\delta)[SI]_C\\
\dfrac{d[SI]_H}{dt} &=  - \beta_C\kappa_{CH}\dfrac{[SI]_H[SI]_C}{S}-\beta_H\kappa_{HH}\dfrac{[SI]^2_H}{S}-(\beta_H+\gamma)[SI]_H+\eta[\SI]_{H}\\
\dfrac{d[\SI]_H}{dt} &= \beta_C\kappa_{CH}\dfrac{[\SS]_H[SI]_C}{S}+\beta_H\kappa_{HH}\dfrac{[\SS]_H[SI]_H}{S}-\beta_C\kappa_{CH}\dfrac{[\SI]_H[SI]_C}{S}-\beta_H\kappa_{HH}\dfrac{[\SI]_H[SI]_H}{S}-(\eta+\gamma)[\SI]_H \\ 
\dfrac{d[SS]_C}{dt} &= -2\beta_C\kappa_{CC}\dfrac{[SS]_C[SI]_C}{S}-2\beta_H\kappa_{CH}\dfrac{[SS]_C[SI]_H}{S} \\
\dfrac{d[\SS]_H}{dt} &= -2\beta_C\kappa_{CH}\dfrac{[\SS]_H[SI]_C}{S}-2\beta_H\kappa_{HH}\dfrac{[\SS]_H[SI]_H}{S}.
\end{aligned} \label{eq:CH}
\end{equation}
As in Section \ref{subsec:pw}, the variables in \eqref{eq:CH} are scaled by $n$ and the initial conditions are given by
\begin{align}
&S(0) = \alpha^S, \quad I(0) = \alpha^I, \quad R(0) = \alpha^R \notag\\%\label{CHic1} \\
&[SI]_C(0) = \mu_C\alpha^S\alpha^I, \quad [SI]_H(0) = 0, \quad [\SI]_H(0) = \mu_H\alpha^S\alpha^I \notag\\
& [SS]_C(0) = \mu_C(\alpha^S)^2, \quad  [SS]_H(0) = 0,\quad [\SS]_H(0)  = \mu_H(\alpha^S)^2. \label{CHic3}
\end{align}
Recall that $\mu_C$ and $\mu_H$ are the average $C$- and $H$-degrees, respectively, of a randomly chosen node and that \eqref{CHic3} corresponds to all community edges being active and all healthcare edges deactivated at $t=0$.  

\blue{Figure \ref{fig:converg} compares trajectories of the stochastic $SIdaR(2,1)$ process compared with system \eqref{eq:CH}.  At smaller graph size $n=10^4$ (Figure \ref{fig:converg}, left), there are discrepancies between the stochastic trajectories and the limiting ODE (red versus blue curves).  As the size of the graph increases, i.e. for $n=10^5$ (Figure \ref{fig:converg}, right), the stochastic trajectories are tightly clustered around the ODE solution, illustrating convergence of the stochastic process to the deterministic system given in Theorem \ref{thm:lln}.}    

\begin{figure}
\includegraphics[scale=0.4]{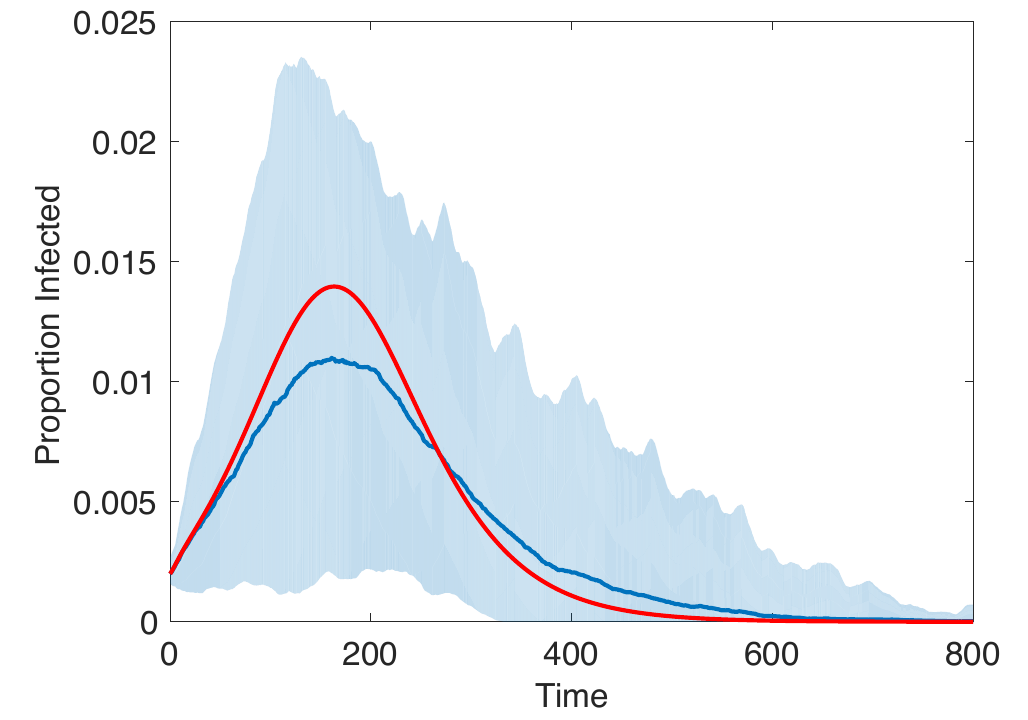}
\includegraphics[scale=0.4]{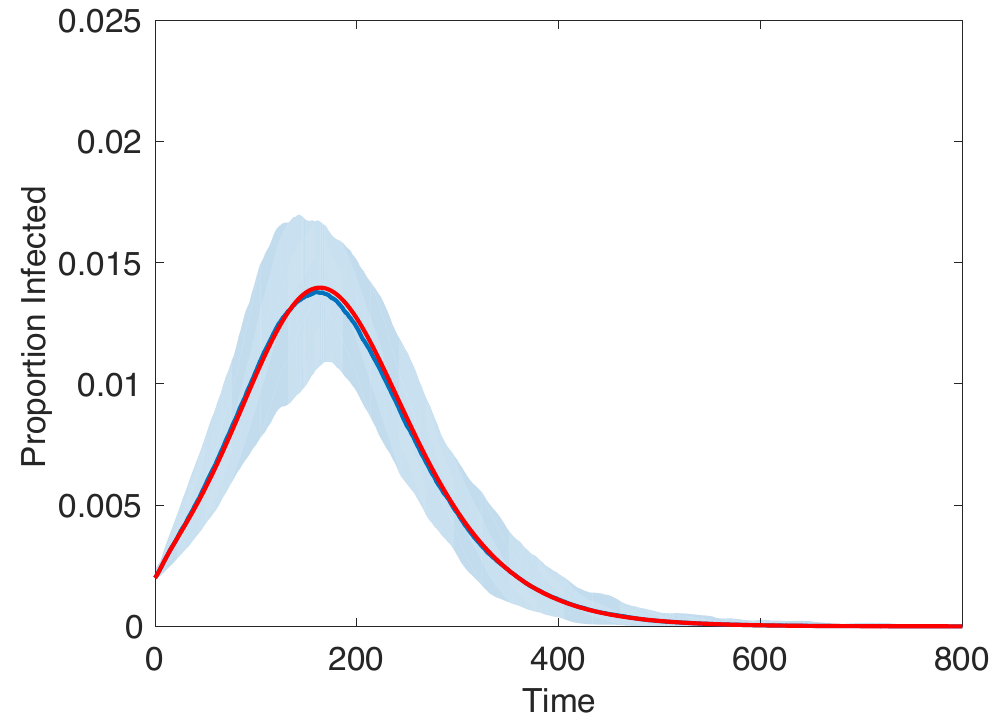}
\caption{\blue{Convergence of the stochastic process for infected to its  large graph limit ($n\to \infty$).  The degree distribution of the $C$-layer is $\operatorname{NB}(5,0.706)$ with $\mu_C = 12$ and $\kappa_{CC} = 1.2$.  The degree distribution of the $H$-layer is regular, i.e. all nodes have equal degree, with $\mu_H = 5$ and $\kappa_{HH} = 0.8$.  The size of the graph is given by  $n=10^4$ (left) and $n=10^5$ (right). The blue shaded regions indicate the approximate 95\% confidence intervals based on 100 numerical simulations of the stochastic $SIdaR(2,1)$ processes, conditioned on the large outbreak.  The blue lines give the mean of the stochastic simulations, and the red lines show the deterministic solutions to \eqref{eq:CH}.  Initially infected nodes ($\alpha^I=0.002$, $\alpha^S=0.998$) are randomly chosen and other parameters are given by $\gamma = 0.08$, $\delta = 0.1$, $\eta = 0.2$, $\beta_C=0.009$, $\beta_H = 0.015$.}}
\label{fig:converg}
\end{figure}

System \eqref{eq:CH} can now be analyzed to gain insight into how the structure of the different layers of the network and their coupling through activation/deactivation of edges in response to infection affects disease dynamics.  In particular, we can compute the basic reproduction number, $\R$, which determines whether disease invasion is possible \cite{vdd2002}.  Consider the disease free equilibrium $\bm{x_0}=(1,0,0,0,0,\mu_C,\mu_H)$.  Let 
\[R_C = \dfrac{\beta_C\kappa_{CC}\mu_C}{\beta_C+\gamma+\delta} \qquad \text{and} \qquad R_H = \dfrac{\beta_H\kappa_{HH}\mu_H\eta}{(\beta_H+\gamma)(\eta+\gamma)}.\]  
Note that $\kappa_{CC}\mu_C = \mu^{ex}_C$, the average excess $C$-degree, so we can interpret $R_C$ as the average number of secondary cases transmitted through the community network in a susceptible population and, likewise, $R_H$ represents the secondary cases caused by healthcare transmission.  The next-generation matrix method \cite{vdd2002} gives (see Appendix \ref{app:R0}) the basic reproduction number as

\begin{subequations}
\begin{alignat}{99}
  \label{eqn:R0_kappas}
 \R  &=& \dfrac{1}{2}\left(R_C+R_H\right) + \dfrac{1}{2}\sqrt{\left(R_C+R_H\right)^2 + 4 R_C R_H \left( \frac{\kappa_{CH}}{\kappa_{CC}} \frac{\kappa_{HC}}{\kappa_{HH}} - 1 \right) } \\
   \label{eqn:R0_cross}
	&&=  \dfrac{1}{2}\left(R_C+R_H\right) + \dfrac{1}{2}\sqrt{  (R_C - R_H)^2 + 4R_{CH}R_{HC}   } ,
\end{alignat}
 \label{eq:R0}
\end{subequations}
where
\begin{equation}
\label{eqn:Rcross}
R_{CH} = \dfrac{\beta_C}{\beta_C + \gamma+\delta}\mu_{C|H}^{ex} \qquad \text{and} \qquad R_{HC} = \dfrac{\beta_H}{\beta_H+\gamma}\dfrac{\eta}{\eta+\gamma}\mu_{H|C}^{ex}. 
\end{equation}
The $R_{CH}$ term can be interpreted as the number of secondary infections created in the community contact network from an initial infection along an $H$ edge, and $R_{HC}$ as the secondary infections created in the healthcare contact network from an initial infection along a $C$ edge. 

\blue{Let us comment about the utility of deriving $\R$ from \eqref{eq:CH}.  The basic reproduction number is a fundamental quantity in epidemiology, and understanding the dependence of $\R$ on model parameters can be helpful for assessing different intervention strategies.  For complex systems such as the stochastic process on a multilayer network considered here, deducing the form for $\R$ on heuristic grounds can be difficult.   Using the next-generation matrix to compute $\R$ from the limiting system of ODEs takes the guesswork out of this process.  As shown in \cite{vdd2002}, the next generation matrix gives a stability criterion for the disease free equilibrium for the limiting ODEs.  Theorem \ref{thm:lln} shows that the threshold for when the branching process for the stochastic model is subcritical converges to the threshold for the limiting ODEs, i.e. when the $\R$ expression is equal to one.} 

\blue{Expression \eqref{eq:R0} is equivalent to the $\R$ expression found by \cite{Ball2008} and \cite{Li2015}, with the difference that the edge activation and deactivation rates enter into the expression.    Indeed,} the presence of both $R_C$ and $R_H$ in $\R$ reflects the coupling of the different layers through edge activation/deactivation.  In the limit of fast deactivation of community contacts (i.e. $\delta \rightarrow \infty$), $\R \rightarrow R_H$ and disease invasibility depends solely upon the healthcare layer.  Similarly, in the limit of slow activation (i.e. $\eta \rightarrow 0$), $\R \rightarrow R_C$ and the basic reproduction number is driven by the community layer.  For intermediate activation and deactivation rates, $\R$ depends upon both layers and the multilayer aspect of the model plays an important role in affecting disease dynamics.  Both transmission ``within" (e.g. $R_C, R_H$ terms) and ``between" layers (e.g. $R_{CH}, R_{HC}$ terms) contribute to $\R$, as seen in \eqref{eqn:R0_cross}.

\blue{An important point is that \eqref{eq:R0} shows that knowledge of the reproduction numbers $R_C, R_H$ of the individual layers is not sufficient to determine $\R$ for the entire network.  Instead, $\R$ depends upon the structure of the different layers through the parameters $\kappa_{ij}$.  }
Consider the cross term in \eqref{eqn:R0_kappas}.  In the case where $\kappa_{CH} \kappa_{HC} / \kappa_{CC} \kappa_{HH} = 1$ \blue{(for example, independent  Poisson degree distributions in each layer)}, the cross term vanishes and $\R = R_C + R_H$.  In general the sign of the cross term may be either positive or negative, and thus $\R$ may be either larger or smaller than $R_C + R_H$.  In fact for fixed $R_C$ and $R_H$ it is possible for $\R$ to be either greater than or less than one, depending upon the structure of the layers.  This is illustrated in Figure \ref{fig:kappa} where we plot prevalence (i.e. infected proportion of the population) over the course of an epidemic for two different scenarios for the structure of the healthcare layer.  \blue{We take  the  two layer-specific  degrees to be  independent} ($\kappa_{CH}=\kappa_{HC}=1$) with  the degree distribution of the community layer being Poisson ($\kappa_{CC}=1$) with $\mu_C=10$.   We fix $R_C = 0.75$ and $R_H=0.5$ and take  the degree distribution of the healthcare layer  as  negative binomial with $\mu_H=8$; in the first scenario (blue) we \blue{take  the healthcare layer degree as  $NB(10,4/9)$ (corresponding to $\kappa_{HH}=1.1$) while in the second (orange) we take the healthcare layer degree as  $NB(1/3,0.96)$ (corresponding to $\kappa_{HH} =4$)}.  Figure \ref{fig:kappa} shows, for each scenario, the deterministic solution to the limiting system \eqref{eq:CH} (solid line) as well as the approximate 95\% empirical confidence interval calculated from 500 stochastic simulations of the corresponding $SIdaR(2,1)$ process with $n=10^4$ (shaded region). The basic reproduction numbers are calculated from \eqref{eqn:R0_kappas} to be, respectively, $\R=1.17$ and $\R=0.91$.  Correspondingly, in the first case a large epidemic occurs while in the second case the initial infection quickly dies out.  The increase in $\kappa_{HH}$ from the first case to the second corresponds to an increase in the variance of the $H$ degree distribution.  In fact, analogous to previous results (see, e.g., \cite{Li2015}), inspection of \eqref{eqn:R0_cross} reveals that, if $\mu_C$ and $\mu_H$ are kept constant, $\R$ is an increasing function of the variances of the $C$ and $H$ degree distributions as well as their covariance.  

\begin{figure}
\includegraphics[scale=0.5]{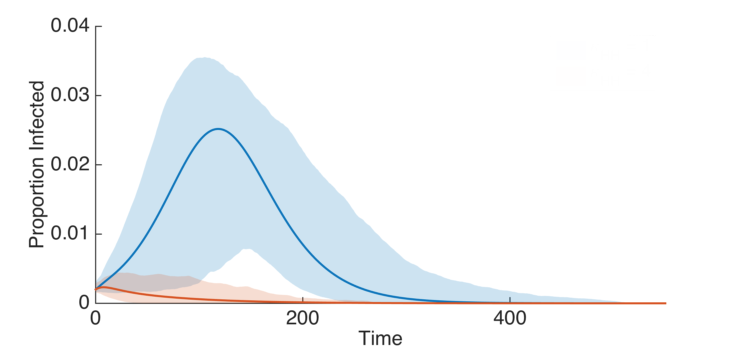}
\caption{Prevalence curves for two different scenarios for the structure of a community-healthcare network  corresponding to large (top/blue) and small (bottom/orange) outbreaks.  The degree distribution of the $C$-layer is $\operatorname{Pois}(10)$.  The degree distribution of the $H$-layer is $\operatorname{NB}(10,4/9)$ with $\kappa_{HH}=1.1$ in the first case (top/blue) and $\operatorname{NB}(1/3,0.96)$ with $\kappa_{HH}=4$ in the second case (bottom/orange).  The solid lines show the deterministic solutions to \eqref{eq:CH} while the shaded regions indicate the approximate 95\% confidence intervals based on 500 numerical simulations of the stochastic $SIdaR(2,1)$ processes with $n=10^4$.  Initially infected nodes ($\alpha^I=0.002$, $\alpha^S=0.998$) are randomly chosen and $\gamma = 0.1$, $\delta = 0.1$, and $\eta = 0.3$.  We fix $R_C=0.75$ and $R_H=0.5$ which corresponds to $\beta_C=0.0162$ and, respectively for the two scenarios, $\beta_H = 0.0082$ and $\beta_H =0.0021$.}
\label{fig:kappa}
\end{figure}

In many situations $\R$ not only determines the ability of disease to invade, but also the size of an outbreak if one occurs \cite{Ma2006}.  The system of ODEs \eqref{eq:CH} obtained via Theorem \ref{thm:lln} and Corollary \ref{cor:PT} can similarly be analyzed to 
determine a relation for the final size of the epidemic, as in \cite{Andersson2000,Arino2007}.  To illustrate, consider the special case where the degrees in community and healthcare layers are independent with Poisson degree distributions (i.e. $\kappa_{CC}=\kappa_{HH}=\kappa_{CH} = 1$).  Let $S_\infty$ denote the fraction of the population that escapes infection.  Analysis of a transformed model (as in \cite{Andersson2000}) or, alternatively, a result of Arino et al. \cite{Arino2007} can then be used (see Appendix \ref{app:final}) to derive the final size relation: 
\begin{equation}
\log\left(\frac{S_\infty}{\alpha^S}\right) = -\R\left(\alpha^S+\alpha^I-S_\infty\right) \label{CHfinal}
\end{equation}
which agrees with the classic result for mean-field, homogeneously mixed populations as $\alpha^I \to 0$ and $\alpha^S \to 1$ \cite{kermack1927}.

We conclude our analysis of the community-healthcare network model by noting that system \eqref{eq:CH} is simple enough to be amenable for practical application to outbreaks of interest, for example for parameter estimation and intervention assessment from available data (for example, medical records and contact tracing data).  \blue{For  a given set of parameters  system \eqref{eq:CH} can be used to assess impact of different interventions on preventing outbreaks from occurring (e.g. bringing $\R < 1$) or decreasing the size of an outbreak if it does occur (e.g. by using \eqref{CHfinal} to compute the sensitivity of $S_\infty$ to parameters $p$).}
Further details on statistical inference and application to specific outbreaks will be presented elsewhere.  Here we briefly show how system \eqref{eq:CH} can be further reduced by finding invariants that allow for dimension reduction.  Details for the following methods are provided in Appendix \ref{app:inv}.

Consider the case where $\kappa_{CH} \neq \kappa_{CC}$.  We can then eliminate $[\SS]_H$ by expressing it as a function of $S$ and $[SS]_C$:
\begin{align}
[\SS]_H = \dfrac{\mu_H(\alpha^S)^{2(1-\sigma-\lambda)}}{\mu_C^{\sigma}}[SS]_C^{\sigma}S^{2\lambda} \label{inv1}
\end{align}
where 
\begin{equation}
\sigma = \dfrac{\kappa_{HH}-\kappa_{CH}}{\kappa_{CH}-\kappa_{CC}} \qquad \text{and} \qquad \lambda = \kappa_{CH}-\sigma\kappa_{CC} = \dfrac{\kappa_{CH}^2-\kappa_{CC}\kappa_{HH}}{\kappa_{CH}-\kappa_{CC}} \label{eq:sigma}.
\end{equation}

In the case of independent degrees in different layers ($\kappa_{CH}=1$) neither of which has a Poisson degree distribution (i.e. $\kappa_{CC} \neq 1$ and $\kappa_{HH} \neq 1$), we are able to find two additional invariants, equations \eqref{inv2} and \eqref{inv3} below. The reduced system, of dimension four, is given by
\begin{equation}
\begin{aligned}
\dfrac{dS}{dt} &= -\beta_H[SI]_H-\beta_C[SI]_C \label{SCHeqn} \\
\dfrac{dI}{dt} &= \beta_H[SI]_H+\beta_C[SI]_C-\gamma I \\
\dfrac{d[SI]_H}{dt} &=  - \beta_C\dfrac{[SI]_H[SI]_C}{S}-\beta_H\kappa_{HH}\dfrac{[SI]^2_H}{S}-(\beta_H+\gamma)[SI]_H+\eta[\SI]_H\\
\dfrac{d[SS]_C}{dt} &= -2\beta_C\kappa_{CC}\dfrac{[SS]_C[SI]_C}{S}-2\beta_H\dfrac{[SS]_C[SI]_H}{S}
\end{aligned}
\end{equation}
where $[\SS]_H$ is given by \eqref{inv1} and 
\begin{align} 
[SI]_C &= \dfrac{\beta_C+\gamma+\delta}{\beta_C(1-\kappa_{CC})}S - [SS]_C+\left(\dfrac{(\beta_C+\gamma+\delta)}{\beta_C(\kappa_{CC}-1)}+\mu_C(\alpha^S+\alpha^I)\right)\left(\dfrac{[SS]_C}{\mu_C}\right)^{1/2} \label{inv2} \\
[\SI]_H &= \dfrac{(\beta_H+\gamma)(\eta+\gamma)}{\eta\beta_H(1-\kappa_{HH})}S - \dfrac{\eta+\gamma}{\eta}[SI]_H - [\SS]_H +\left(\dfrac{(\beta_H+\gamma)(\eta+\gamma)}{\eta\beta_H(\kappa_{HH}-1)} + \mu_H(\alpha^S+\alpha^I)\right)\left(\dfrac{[\SS]_H}{\mu_H}\right)^{1/2}. \label{inv3}
\end{align}

In the case of independent degree layers with Poisson degree distributions ($\kappa_{CC}=\kappa_{HH}=\kappa_{CH} = 1$, see Appendix \ref{app:final}), we have $[SS]_C=\mu_CS^2$ and $[\SS]_H=\mu_HS^2$.  We can further reduce the dimension of the system by one with the following invariant:
\begin{equation}
\log\left(\frac{S}{\alpha^S}\right) = -\R\left(\alpha^S+\alpha^I-S\right)+\dfrac{\beta_C}{\beta_C+\gamma+\delta}\dfrac{[SI]_C}{S} +\dfrac{\beta_H}{\beta_H+\gamma}\dfrac{[SI]_H}{S} +\dfrac{\beta_H\eta}{(\beta_H+\gamma)(\eta+\gamma)}\dfrac{[\SI]_H}{S}. \label{CHinv}
\end{equation}

\section{Discussion} 
	\label{sec:discuss}

\blue{The complexity of dynamic multilayer networks makes understanding the disease dynamics evolution on such structures  a challenge. Working with simplified  models, which are nevertheless capable of  retaining  most important  aspects  of network evolution and disease transmission, is essential for gaining biological insight into mechanisms underlying basic disease features such as invasion, persistence, and outbreak size.  In this work, we have developed a framework for modeling infectious diseases with multiple modes of transmission in the setting where the network changes in response to infection.}    Even in this seemingly complicated scenario, it is relatively straightforward to formulate a continuous-time stochastic process by considering transitions in the states of nodes and connected pairs of nodes.  \blue{However, the state space of the Markov process becomes unmanageable as the size of the network increases and analysis of the non-Markovian aggregate process is complicated.  Our main result, Theorem~\ref{thm:lln}, rigorously derives the large graph limit of the stochastic process on a layered configuration model graph and, thus, gives a simple model retaining key features of the epidemic process while being amenable to analysis.  Moreover, our results extend previous ones for the SIR process on a static, single-layer \red{configuration model} network \cite{Janson2014,Decreusefond2012}.}   

\blue{As in \cite{Janson2014,Decreusefond2012}, proof of Theorem \ref{thm:lln}, and expression of the limiting system in the general case, require   introduction of the  edge-based variable ($\Xth$).  However, in contrast to previous results, we have defined the limiting system here in terms of dyads via the function $\bar{\kappa}$, which also turns out to provide the large-graph-consistent approximation of triples.    Consequently,  we obtain a simple characterization of the class of degree distributions on random configuration model graphs  for which the model obtained via ordinary pair approximation \cite{Rand1999} coincides with the limiting system described by Theorem~\ref{thm:lln}.  \blue{The characterization criterion  may be formulated  in terms of  average and excess degrees ($\bar{\kappa}\equiv \kappa$) or, equivalently, in terms of the corresponding identity for the probability generating function.   For the single layer case,  this condition is satisfied by Poisson but also by binomial and negative binomial distributions.  In the non-PT degree distribution case, as our example in Section~\ref{sec:correct} illustrates}, numerical comparison of $\kappa$ and $\bar{\kappa}$ could potentially be used to quantify how well a pairwise model accurately reflects the relevant disease dynamics.}

\red{Evolving network structure in response to illness is a basic feature that is relevant for many diseases, including Ebola, which was the original motivation for this study.  The importance of transmission both within the community and to care-givers for Ebola has empirical support \cite{dowell1999,kratz2015,roels1999}, and these different transmission routes have been incorporated into modeling studies \cite{Drake2015,Gomes2014,Legrand2007,Merler2015,Pandey2014,Rivers2014,Tsanou2015,Valdez2015,Webb2014}.  
Funeral transmission of Ebola is additionally an important concern \cite{hewlett2003,maganga2014}, and many modeling studies have incorporated Ebola transmission through unsafe burial practices \cite{Legrand2007,Webb2014,Weitz2015}.  The general multi-layer framework we have presented here with activation and de-activation of edges is flexible and can be extended to include these different Ebola-specific considerations.  For example, we can consider a layer corresponding to disease transmission through funerals, with activation of edges in this layer occurring following entry into a death via disease class.  Similarly, the basic SIR states for the nodes presented here can be extended to incorporate additional states such an exposed (latent) class, which is important for Ebola \cite{Lekone2006}.  The 2014 West Africa Ebola outbreak prompted an outpouring of modeling studies, utilizing a variety of approaches including ODEs \cite{Fisman2014,Rivers2014,Tsanou2015}, stochastic processes \cite{Drake2015,Pandey2014}, metapopulations \cite{Merler2015,Valdez2015}, and contact networks \cite{Scarpino2014,Wells2015}.  Our work contributes a rigorous approach to understanding how the evolution of network structure in response to infection impacts disease dynamics on layered configuration model networks.}

\red{The modeling framework we present is flexible, for example being able accommodate arbitrary joint degree distributions with finite second moment for an arbitrary number of layers.   As in the single layer configuration model, however, wiring within each layer occurs at random.  Empirical networks often possess features such as community structure \cite{porter2009} or triad closure \cite{watts1998} that are not present in configuration model networks.  Despite this, the layered configuration model setting can still yield biological and epidemiological insights due to its tractability for analysis and statistical inference.  }
As demonstrated in the two-layer community-healthcare example in Section \ref{sec:CH}, the large graph limit we have derived is indeed tractable.  The basic reproduction number can be calculated, and its analysis provides insight into how the structure of the different layers and their coupling through edge dynamics affect disease invasion and the final size of the outbreak.  \red{Theorem \ref{thm:lln} and Corollary \ref{cor:PT} also provide reference points for examining the impact of more realistic network features such as clustering and community structure on disease spread on dynamic multilayer networks.}

\red{The general framework presented here} can be adapted to diseases with multiple modes of transmission (e.g. those with sexually and non-sexually transmitted infections).  Application of this framework to specific diseases, such as Ebola, will require investigation into parameter identifiability and statistical estimation methods.  The large graph limit derived here will aid such analysis and, in fact, suggests a hybrid approach in which the node state transitions remain stochastic but the dyads are approximated using the limiting differential equations (or, if possible, invariants such as those found in Section \ref{sec:CH}).  Even for large networks, the resulting Markov process approximation would allow for computationally inexpensive maximum likelihood estimation.  Finally, we mention that interventions (e.g. vaccination) can also be incorporated into this framework.  \blue{The relative simplicity of the limiting system  should  allow for evaluation of the impact of proposed interventions, for example via sensitivity analysis of $\R$ and the final outbreak size or via  methods of optimal control (e.g. optimize  vaccine distribution, see \cite{Lenhart2007}).  This could be critical for providing actionable recommendations to public health policymakers with the aim of curbing current epidemics or preventing future outbreaks.}

\appendix
\section{\blue{Summary of notation}}
	\label{app:notation}

\blue{The notation for the layered configuration model graph:}\\

\blue{
\begin{tabular}{ll}
$n$ & number of nodes \\
$r$ & number of layers in layered configuration model \\
$\bk = (k_1,...,k_r)$ & multi-degree with degree $k_i$ in layer $i$ \\
$\psi$ & probability generating function of the multivariate degree distribution \\
$\G_r(\psi,n)$ & layered configuration model with $n$ nodes, $r$ layers and degree distribution $\psi$ \\
$\mu_i$ & average $i$-degree \\
$\kappa_{ji}$ & normalized average excess $j$-degree of an $i$-neighbor \\
$\mu_i^S$ & average excess $i$-degree of a susceptible node \\
$\mu_{i|j}^{ex|SI}$ & average excess $i$-degree of a susceptible node chosen randomly \\ & as a $j$-neighbor of an infectious node \\
\end{tabular}\\[0.1in]
}

\blue{The notation for the $SIdaR(r,m)$ process:\\}

\blue{
\begin{tabular}{ll}
$m$ & number of deactivating layers \\
$\beta_j$ & rate of infection along $j$-edges \\
$\delta_j$ & rate of deactivation (drop) of $j$-edges \\
$\eta_j$ & rate of activation of $j$-edges \\
$\gamma$ & rate of recovery \\
$X^{SI,u}_j$, $X^{SS,u}_j$ & number of infectious, susceptible active $j$-neighbors of $u$, for susceptible $u$ \\
$X^{\SI,u}_j$, $X^{\SS,u}_j$ & number of infectious, susceptible deactivated $j$-neighbors of $u$, for susceptible $u$ \\
$X^S$, $X^I$, $X^R$ & number of nodes (of given disease status) \\
$\XSI$, $\XSS$ & number of active edges between nodes (of given status) \\
$\XSIT$, $\XSST$ & number of deactive edges between nodes (of given status) \\
$\XSdot$ & number of edges belonging to susceptible nodes \\
$\Xth$ & quantity defined by equation \eqref{eq:th} \\
$\bm{\F}^\X$, $\bm{\F}^\Xth$ & integrable function part of evolution of $\X$, $\Xth$ as in equations \eqref{eq:X}, \eqref{eq:intth} 
\end{tabular}\\[0.1in]
}

\blue{The notation for the large graph limit:\\}

\blue{
\begin{tabular}{ll}
$S$, $I$, $R$ & number of nodes (of given disease status) \\
$[SI]$, $[SS]$ & number of active edges between nodes (of given status) \\
$[\SI]$, $[\SS]$ & number of deactive edges between nodes (of given status) \\
$\th$ & large graph limit of $\Xth$ \\
$\bm{\alpha} = (\alpha^S$,...,$\alpha^{\SS})$ & initial condition of $\bx = (S,...,[\SS])$, scaled by $n$ \\
$\bm{\H}^\X$, $\bm{\H}^\th$ & integrable function for evolution of $\X$, $\X^\th$ in the large graph limit as in equation \eqref{eq:H} \\
$\bar{\kappa}_{jl}$ & function of network structure used in large graph limit, defined by equation \eqref{eq:kappabar} 
\end{tabular}
} 	
\section{Proof of the limit theorem} 
	\label{sec:proof}

In this section, we provide the proof of our main result, Theorem~\ref{thm:lln}, preceded by two lemmas.  The derivation of our results relies on the key observation, summarized in Remark \ref{rem:hg}, that in a finite graph the neighborhood of a susceptible node may be described by a certain multivariate hypergeometric distribution.  Lemma \ref{lem:XSdot} shows that $X^{S_\bk}$ and $\XSdot$ can be expressed in the limit as functions of $\Xth$ given by \eqref{eq:th}.  Lemma \ref{lem:Delta} shows that the dynamics of the scaled process on the finite graph converges, in the appropriate sense, to the dynamics described by the ODE system \eqref{eq:D} involving $\th$.
 Theorem \ref{thm:lln} then follows from Lemma \ref{lem:Delta} using Doob's  and Gronwall's inequalities.  

Recall that we take operations on vectors such as multiplication, division, integration and ordering to be coordinatewise.  We first provide an important remark about the layer $j$ neighborhood of a susceptible node of degree $\bk$, i.e. $(X_j^{SI,i}(t),X_j^{\SI,i}(t),X_j^{SS,i}(t),X_j^{\SS,i}(t))$ for $i \in S_\bk$.  The distribution of the neighborhood is critical when we consider the expectations of 
\[Q^{jl,(\bk)}_{i} = X_l^{SI,i}(X_j^{SS,i}-X_j^{SI,i}) \quad \text{and} \quad \tilde{Q}^{jl,(\bk)}_{i} = X_l^{SI,i}(X_j^{\SS,i}-X_j^{\SI,i}),\]
which are mixed moments with respect to the neighborhood distribution. Recall that we consider the evolution of the $SIdaR(r,m)$ process on a realization of the LCM random graph that has been generated by time $t=0$.  However, we could alternatively consider an equivalent process whereby the graph is revealed dynamically as infections occur, as in \cite{Decreusefond2012,Janson2014}. \blue{ In the latter process, a {\em susceptible node} $i \in S_\bk$ remains unpaired until it becomes infected.  Equivalently we could also  pair off all unpaired edges at any time $t>0$ (uniformly at random and independently in each layer,  according to the LCM construction) in order to define the neighborhood of node $i$.  The following remark is perhaps most easily understood by keeping this equivalent construction in mind and recalling that the  probability space  considered is that of all  random configurations, as described in Section \ref{sec:gen}.}

\begin{remark} \label{rem:hg}
For  $i\in S_\bk$ and $k_j \geq 1$,  and conditionally on the process  history up to time $t$,   the vector  $(X_j^{SI,i}(t),X_j^{\SI,i}(t),X_j^{SS,i}(t),X_j^{\SS,i}(t))$ follows the multivariate hypergeometric distribution with probability mass function 
\begin{align*}
P(X_j^{SI,i}(t)&=n_j^{SI},X_j^{\SI,i}(t)=n_j^{\SI},X_j^{SS,i}(t)=n_j^{SS},X_j^{\SS,i}(t)=n_j^{\SS}) \\
&=\frac{\binom{X_j^{SI}(t)}{n_j^{SI}}\binom{X_j^{SS}(t)}{n_j^{SS}}\binom{X_j^{\SI}(t)}{n_j^{\SI}}\binom{X_j^{\SS}(t)}{n_j^{\SS}}\binom{X_j^{S\bigcdot}(t)-X_j^{SI}(t)-X_j^{SS}(t)-X_j^{\SI}(t) - X_j^{\SS}(t)}{k_j-n_j^{SI}-n_j^{SS}-n_j^{\SI}-n_j^{\SS}}} {\binom{X_j^{S\bigcdot}(t)}{k_j}},
\end{align*}
 supported on the four-simplex $0\le n_j^{SI} + n_j^{\SI} + n_j^{SS} +n_j^{\SS}\le k_j$.  This implies 
 \[E_h[X_j^{SI,i}] = k_jX_j^{SI}/X_j^{S\b}\]
 where $E_h$ denotes expectation with respect to the hypergeometric distribution.  We also note that, \blue{based on the LCM construction}, the neighborhoods of $i$ in distinct layers are independent, i.e. 
 \[P(\bm{X^{SI,i}}(t)=\bm{n^{SI}},\bm{X^{SS,i}}(t)=\bm{n^{SS}}) = \prod_{l=1}^r P(X_l^{SI,i}(t)=n_l^{SI},X_l^{SS,i}(t)=n_l^{SS}).\]

It follows from the above that the mixed moments are given by
 \begin{equation}
  E_h\left[Q^{jl,(\bk)}_{i}\right] = k_lk_j \dfrac{X_l^{SI}(t)}{X_l^{S\b}(t)}\left(\dfrac{X_j^{SS}(t)}{X_j^{S\b}(t)}-\dfrac{X_j^{SI}(t)}{X_j^{S\b}(t)}\right) \label{eq:EhQ}
  \end{equation}
  for $l \neq j$ and 
   \begin{equation}
  E_h\left[Q^{jj,(\bk)}_{i}\right] = k_j(k_j-1) \dfrac{X_j^{SI}(t)}{X_j^{S\b}(t)}\left(\frac{X_j^{SS}(t)}{X_j^{S\b}(t)-1}-\frac{X_j^{SI}(t)-1}{X_j^{S\b}(t)-1}\right) -k_j \dfrac{X_j^{SI}(t)}{X_j^{S\b}(t)}. \label{eq:EhQjj}
  \end{equation}
  Likewise, for any $1 \le j,l \le r$,
 \begin{equation*}
  E_h\left[\tilde{Q}^{jl,(\bk)}_{i}\right] = k_lk_j \dfrac{X_l^{SI}(t)}{X_l^{S\b}(t)}\left(\dfrac{X_j^{\SS}(t)}{X_j^{S\b}(t)}-\dfrac{X_j^{\SI}(t)}{X_j^{S\b}(t)}\right).
  \end{equation*}
 \end{remark}
 
Keeping in mind this remark, we proceed to prove the first lemma, which shows that $X^{S_\bk}$ and $\XSdot$ can be expressed in the limit as functions of $\Xth$.

\begin{lemma}
\label{lem:XSdot} 
Assume \ref{A1} and $\sum_\bk ||\bk|| p_\bk <\infty$.  Then,
\begin{enumerate}[label=(\alph*)]
\item $\sup_{0<t\le T}|n^{-1}X^{S_\bk}-\alpha^Sp_\bk\Xth^\bk |\stackrel{P}{\rightarrow}0$ for any $\bk \geq 0$, and \label{lem:fixk}
\item $\sup_{0<t\le T} \norml n^{-1}\XSdot(t)-\alpha^S\Xth(t)\bm{\p\psi}(\Xth(t)) \normr \stackrel{P}{\rightarrow}0$. \label{lem:thpth} 
\end{enumerate}
\end{lemma}

\begin{proof}
$(a)$. Note $X^{S_{\bk}}(t)=\sum_ {i=1}^n Z_i(t)$ where $Z_i(t)\in \{0,1\}$ indicates whether node $i$ is of degree $\bk$ and susceptible at time $t>0$. Recall from Remark \ref{rem:hg} that $EX_j^{SI,i} = k_jX_j^{SI}/X_j^{S\b}$ for $i \in S_\bk$.  We claim that 
 \begin{equation}
 EZ_i(t)=P(i \in S_\bk(t)) = n^{-1}S(0)p_\bk\Xth^\bk. \label{eq:EZi}
 \end{equation}
    Indeed, $EZ_i(t) = P(Z_i(t)=1) = P(i \in S_\bk(t)) = P(i \in S_\bk(t) | i \in S_\bk(0))P(i \in S_\bk(0))$
  where $P(i \in S_\bk(0))  = p_\bk n^{-1}S(0)$ by \ref{A3} and 
 \begin{equation*}
 P(i \in S_\bk(t) | i \in S_\bk(0))  = \exp \left(-\sum_{j=1}^r \beta_j \int_0^t E_h\left[ X_j^{SI,i}(s)\right]ds\right) =\exp\left(-\sum_{j=1}^r \beta_j \int_0^t k_j\dfrac{X_j^{SI}(s)}{X_j^{S\b}(s)}ds\right) =\Xth^\bk.
 \end{equation*}
Equation \eqref{eq:EZi} then implies that $\{X^{S_{\bk}}(t)-S(0)p_{\bk}\Xth^{\bk}(t)\}_{t\ge 0}$ is a c\`{a}dl\`{a}g martingale process with mean zero and finite variation.  By the triangle inequality, 
\begin{align*}
P & \left(\sup_{0<t\le T}\left |n^{-1}X^{S_{\bk}}(t)-\alpha^Sp_{\bk}\Xth^{\bk}(t)\right|>\ve\right) \\
& \le P \left(\sup_{0<t\le T}\left|n^{-1}X^{S_{\bk}}(t)-n^{-1}S(0)p_{\bk}\Xth^{\bk}(t)\right|>\frac{\ve}{2}\right) + P \left(\sup_{0<t\le T}\left|n^{-1}S(0)p_{\bk}\Xth^{\bk}(t)-\alpha^Sp_{\bk}\Xth^{\bk}(t)\right|>\frac{\ve}{2}\right).
\end{align*} 
The second term tends to zero by assumption \ref{A3} and for the first term we have, by Doob's martingale inequality, 
\begin{equation*}
P \left(\sup_{0<t\le T}\left|n^{-1}X^{S_{\bk}}(t)-n^{-1}S(0)p_{\bk}\Xth^{\bk}(t)\right|>\frac{\ve}{2}\right) \le \left(\dfrac{\ve}{2}\right)^{-2}n^{-2}\,Var\, X^{S_{\bk}}(T).
\end{equation*}
Since there are at most $n$ jumps for $X^{S_\bk}$ and each is of size one, it follows that the quadratic variation of $X^{S_{\bk}} = O(n)$, which gives the needed result.  
   
\paragraph*{$(b).$} 
By equation \eqref{eq:XSdot}, we can write 
\[n^{-1}\XSdot(t) - \alpha^S\Xth\bm{\p\psi}(\Xth(t)) = \sum_\bk \bk n^{-1}X^{S_\bk}(t) - \alpha^S \sum_\bk \bk p_\bk\Xth^\bk(t).\] 
Consider arbitrarily  large $N$ and $\ve >0$. By Markov's inequality, we have   
\begin{align*}
P & \left(\sup_n\sup_{0<t\le T}  \norml \sum_{||\bk|| > N} \bk n^{-1}X^{S_{\bk}}(t)- \alpha^S\bk p_{\bk} \Xth^{\bk}(t)\normr >\ve \right) \\
&\le \ve^{-1} \sum_{||\bk|| > N} ||\bk|| \sup_n\sup_{0<t\le T}E\left|n^{-1}X^{S_{\bk}}(t)- \alpha^Sp_{\bk} \Xth^{\bk}(t)\right| \nonumber \\ 
&  \le 3\ve^{-1} \sum_{||\bk|| > N} ||\bk|| p_{\bk}
\end{align*}
since $EX^{S_{\bk}}/n\le 2p_{\bk}$ for $n$ sufficiently large, $\alpha^S \le 1$, and we may apply the Monotone Convergence Theorem. So the tail of the sum is negligible since $N$ is arbitrary and, by assumption, $\sum_\bk ||\bk|| p_\bk <\infty$.  The result then follows since in $(a)$ we showed convergence for fixed $\bk$. 
\end{proof}

Before proceeding with the next lemma, we give a brief remark on boundedness of our variables and define some useful functions.

\begin{remark} \label{rem:bdd}
We note that $\alpha^S\Xth\bm{\p\psi}(\Xth) \le \bm{\p\psi}(\mathbf{1})$ and, for sufficiently large $n$, $n^{-1}\XSI \le n^{-1}\XSdot \le 2\bm{\p\psi}(\mathbf{1})$ (and likewise for $n^{-1}\XSS$).  By \ref{A1}, $n^{-1}\XSdot$ is  bounded away from 0 on $[0,T]$ for finite $T$ and, thus so is $\Xth$. Furthermore, by Lemma \ref{lem:XSdot}\ref{lem:thpth}, we can take the same lower bound for $\alpha^S\Xth\bm{\p\psi}(\Xth)$.  Let $\xi>0$ be a uniform lower bound for $n^{-1}\XSdot$, $\Xth$ and $\alpha^S\Xth\bm{\p\psi}(\Xth)$.  We will use the notation $[\bm{\xi},2\bm{\p\psi}(\one)]^r := [\xi_1,2\p_1\psi(\one)]\times \hdots \times [\xi_r,2\p_r\psi(\one)]$ and hence may write  $\XSdot/n \in [\bm{\xi},2\bm{\p\psi}(\one)]^r$.
\end{remark}

Let $\bm{\F}^\X$ and $\bm{\F}^\Xth$ be defined as in \eqref{Feqn} and \eqref{eq:Ftheta} and $\bm{\H}$ as in \eqref{eq:H}. Define $\bm{\Delta}(t) = (\bm{\Delta}^\X,\bm{\Delta}^\Xth)(t)$ where $\bm{\Delta}^\X = (\Delta^S,\Delta^I,\bm{\Delta^{SI}},\bm{\Delta^{\SI}},\bm{\Delta^{SS}},\bm{\Delta^{\SS}})$ and $\bm{\Delta}^\Xth$ are given by
\begin{equation}
\bm{\Delta}^\X(t) = n^{-1}\bm{\F}^\X(\X(t)) - \bm{\H}^\X(n^{-1}\X(t),\Xth(t)) \label{eq:Delta}
\end{equation}
and
\begin{equation}
\bm{\Delta}^\Xth(t) = \bm{\F}^\Xth(\XSI(t),\XSdot(t),\Xth(t)) - \bm{\H}^\Th(n^{-1}\XSI(t),\Xth(t)). \label{eq:Deltath}
\end{equation}
Lemma \ref{lem:Delta} shows that $\bm{\Delta}^\X$ and $\bm{\Delta}^\Xth$ tend to zero uniformly in probability.  
The convergence of $\bm{\Delta}^\Xth$ to zero will follow easily from Lemma \ref{lem:XSdot}.  The convergence of $\bm{\Delta}^\X$ to zero is less obvious  and its proof involves consideration of the empirical hypergeometric mixed moments, i.e. $\sum_{i \in S_k}Q^{jl,(\bk)}_i$ and $\sum_{i \in S_k}\tilde{Q}^{jl,(\bk)}_i$.  We will use the facts that,  in the limit, the hypergeometric mixed moments are approximately multinomial and we can replace $\XSdot$ with a function of $\Xth$ by Lemma \ref{lem:XSdot}.  Therefore, it will be convenient to define the following compensators. 
   
Let $C_{h}^{jl,(\bk)}: [0,T]\to \bR$ be given by 
 \begin{equation}
  C_{h}^{jl,(\bk)}(t) = k_lk_j \dfrac{X_l^{SI}(t)}{X_l^{S\b}(t)}\left(\dfrac{X_j^{SS}(t)}{X_j^{S\b}(t)}-\dfrac{X_j^{SI}(t)}{X_j^{S\b}(t)}\right), \qquad l \neq j, \label{eq:Ch}
  \end{equation}
and 
   \begin{equation}
C_{h}^{jj,(\bk)}(t) = k_j(k_j-1) \dfrac{X_j^{SI}(t)}{X_j^{S\b}(t)}\left(\frac{X_j^{SS}(t)}{X_j^{S\b}(t)-1}-\frac{X_j^{SI}(t)-1}{X_j^{S\b}(t)-1}\right) -k_j \dfrac{X_j^{SI}(t)}{X_j^{S\b}(t)}, \label{eq:Chjj}
  \end{equation}
so that the hypergeometric mixed moment in equation \eqref{eq:EhQ} is given by
\begin{equation}
E_h\left[Q^{jl,(\bk)}_i \right] = C^{jl,(\bk)}_{h}(t). \label{eq:EhCh}
\end{equation}

We also define a function related to the multinomial distribution, $C_{m}^{jl,(\bk)}: [0,T] \times[\bm{\xi},2\bm{\partial\psi}(1)]^r \to \bR$, which is given by
 \begin{equation}
  C_{m}^{jl,(\bk)}(t,\bm{z}(t)) = k_jk_l \dfrac{n^{-2}X_l^{SI}(t)}{z_l(t)}\left(\dfrac{X_j^{SS}(t)}{z_j(t)}-\dfrac{X_j^{SI}(t)}{z_j(t)}\right), \qquad l \neq j, \label{eq:Ch}
  \end{equation}
and
\begin{equation}
C_{m}^{jj,(\bk)}(t,\bm{z}(t)) = k_j(k_j-1)\dfrac{n^{-2}X_j^{SI}}{z_j(t)}\left(\dfrac{X_j^{SS}}{z_j(t)}-\dfrac{X_j^{SI}}{z_j(t)}\right)-k_j\dfrac{n^{-1}X_j^{SI}}{z_j(t)}. \label{eq:EmQjj}
\end{equation}
(Note that  if quantities  $n^{-1}X_j^{SI}/z_j(t)$ and $n^{-1}X_j^{SS}/z_j(t)$ were actual probabilities then $C_m^{jl,(\bk)}(t,\bm{z}(t))$ would be a mixed moment of  the multinomial distribution.)
Observe that there exists $L>0$ such that 
\begin{equation}
C_{m}^{jl,(\bk)} \leq L||\bk||^2 \label{eq:Cbdd}
\end{equation}
for any $j,l$ (and uniformly in $n$) sinsce the domain of $z$ is bounded away from 0 and $n^{-1}\XSI$, $n^{-1}\XSS$ are uniformly bounded above by Remark \ref{rem:bdd}.  It also follows that $C_m^{jl,(\bk)}$ is Lipschitz continuous in $\bm{z}$.

\begin{lemma}\label{lem:Delta} 
Assume $\ref{A1}-\ref{A2}$.  Then, 
\begin{enumerate}[label=(\alph*)]
\item $\sup_{0<t\le T} ||\bm{\Delta}^\Xth(t)|| \stackrel{P}{\rightarrow}0$, and
\item $\sup_{0<t\le T} ||\bm{\Delta}^\X(t)|| \stackrel{P}{\rightarrow}0.$
\end{enumerate}
\end{lemma}

\begin{proof} $(a)$.
Define $\bm{\J}: [0,T]\times [\bm{\xi},2\bm{\p\psi}(\one)]^r \to \bR^r$ given by $\bm{\J}(t,\bm{z}(t)) = -\beta\Xth(t)n^{-1}\XSI(t)/\bm{z}(t)$.
By Remark \ref{rem:bdd} and \ref{A1}, $(t,n^{-1}\XSdot)$ and $(t,\alpha^S\Xth\bm{\p\psi}(\Xth))$ are in the domain of $\bm{\J}$ for $t \in [0,T]$. 
By definition of $\bm{\Delta}^\Xth$ in \eqref{eq:Deltath},
\begin{equation*}
\bm{\Delta}^\Xth(t) = -\beta \dfrac{\XSI(t)\Xth(t)}{\XSdot(t)} + \beta\dfrac{n^{-1}\XSI(t)}{\alpha^S\bm{\p\psi}(\Xth(t))} = \bm{\J}(t,n^{-1}\XSdot(t))-\bm{\J}(t,\alpha^S\Xth(t)\bm{\p\psi}(\Xth(t))).
\end{equation*}
Since $\bm{\J}$ is Lipschitz continuous in $\bm{z}$, 
\begin{align*}
\sup_{0<t\le T} ||\bm{\Delta}^\Xth(t)|| &= \sup_{0<t\le T} ||\bm{\J}(t,n^{-1}\XSdot(t))-\bm{\J}(t,\alpha^S\Xth(t)\bm{\p\psi}(\Xth(t))) || \\
& \le L_1 \sup_{0<t\le T} ||n^{-1}\XSdot(t) - \alpha^S\Xth(t)\bm{\p\psi}(\Xth(t))||
\end{align*}
for some $L_1>0$.  The result then follows from Lemma \ref{lem:XSdot}\ref{lem:thpth} since \ref{A2} implies $\sum_\bk ||\bk|| p_\bk < \infty$.

Regarding part $(b)$  note that, by definition in \eqref{eq:Delta}, $\Delta^S = \Delta^I = 0$.  We will show that $\sup_{0<t\le T} ||\bm{\Delta^{SI}}(t)|| \stackrel{P}{\rightarrow}0$ as it follows similarly for $\bm{\Delta^{\SI}}$, $\bm{\Delta^{SS}}$ and $\bm{\Delta^{\SS}}$ and together these imply $\sup_{0<t\le T} ||\bm{\Delta}^\X(t)|| \stackrel{P}{\rightarrow}0$. In fact, we observe that $\bm{\Delta^{SI}} = (\Delta_1^{SI},\hdots,\Delta_r^{SI})$ and so it suffices to show  
 \begin{equation}
 \sup_{0<t\le T} |\Delta_j^{SI}(t)| \stackrel{P}{\rightarrow}0 \label{eq:DeltaSItozero}
\end{equation} 
for $j=1,\hdots,r$.  

Let $1 \le j \le r$.  We can rewrite $\Delta_j^{SI}$ as
 \begin{align*}
\Delta_j^{SI} &= n^{-1}\F_j^{SI}(\XSI,\XSS) - \H_j^{SI}(n^{-1}\XSI,n^{-1}\XSS,\Xth) \\
&= n^{-1}\sum_{i \in S}\left(\sum_{l=1}^r \beta_lX_l^{SI,i} (X_j^{SS,i}-X_j^{SI,i})\right)- \sum_{l=1}^r\left[ \beta_ln^{-2}X_l^{SI} (X_j^{SS}-X_j^{SI})\frac{\p^2_{jl}\psi(\Xth)}{\alpha^S\p_j\psi(\Xth)\p_l\psi(\Xth)}\right]\\ &\qquad +\beta_j n^{-1}X_j^{SI} \\
 &= \sum_\bk \left[ n^{-1}\sum_{i\in S_\bk} \left(\sum_{l=1}^r \beta_lX_l^{SI,i} (X_j^{SS,i}-X_j^{SI,i})\right) - \sum_{l \neq j} \beta_l n^{-2} X_l^{SI}\left(X_j^{SS}-X_j^{SI}\right) \dfrac{ k_jk_lp_\bk\Xth^\bk}{\alpha^S \xth_j \xth_l\p_j\psi(\Xth)\p_l\psi(\Xth)} \non \right. \\
 & \left. \qquad +  \beta_j n^{-2} X_j^{SI}\left(X_j^{SS}-X_j^{SI}\right) \dfrac{ k_j(k_j-1)p_\bk \Xth^\bk}{\alpha^S(\xth_j\p_j\psi(\Xth))^2} + \beta_j n^{-1}X_j^{SI} \dfrac{k_j p_\bk \Xth^\bk}{\xth_j\p_j\psi(\Xth)} \right]\\
 &= \sum_\bk \left[n^{-1}\sum_{i\in S_\bk} \left(\sum_{l=1}^r \beta_lX_l^{SI,i} (X_j^{SS,i}-X_j^{SI,i})\right) - \sum_{l=1}^r\beta_l\alpha^S p_\bk\Xth^\bk C_{m}^{jl,(\bk)}(t,\alpha^S\Xth\bm{\p\psi}(\Xth))\right].
\end{align*}
Thus, we define
\begin{equation*}
\Delta^{SI}_{jl,(\bk)}(t) = n^{-1}\sum_{i\in S_\bk} \left( \beta_lX_l^{SI,i} (X_j^{SS,i}-X_j^{SI,i})\right) - \beta_l\alpha^S p_\bk \Xth^\bk C_{m}^{jl,(\bk)}(t,\alpha^S\Xth\bm{\p\psi}(\Xth))
\end{equation*}
so that $\Delta_j^{SI} = \sum_{l=1}^r \sum_\bk  \Delta^{SI}_{jl,(\bk)}$.  Hence, it suffices to show \begin{equation}\sup_{0<t\le T} |\sum_\bk\Delta^{SI}_{jl,(\bk)}(t)| \stackrel{P}{\rightarrow}0\qquad \text{for  $j,l = 1,\hdots,r$.}\label{eq:Djl}\end{equation} 
This is done in what follows in  two separate steps. Consider an arbitrary pair $(j,l)$  $1\le j,l\le r.$
 We first show that as $N\to \infty$ 
\begin{equation} \label{eq:Deltatail} 
\sup_n\sup_{0<t\le T} \left| \sum_{||\bk||>N}\Delta^{SI}_{jl,(\bk)}(t) \right |  \stackrel{P}{\to}0.
\end{equation}
To this end, observe that  
\begin{align}
& \left| n^{-1}\sum_{||\bk||>N}\sum_{i\in S_\bk} \beta_l  X_l^{SI,i}(t)(X_j^{SS,i}(t)-X_j^{SI,i}(t)) \right| \non \\
& \le n^{-1} \sum_{||\bk||>N}\sum_{i \in S_\bk}| \beta_l k_l k_j | \le  n^{-1}\beta_lC'  \sum_{||\bk||>N} ||\bk||^2 X_{S_k}(t)  \le 2\beta_lC'  \sum_{||\bk||>N} ||\bk||^2 p_k   \label{mmtail}
\end{align}
for some $C'>0$ since $X_{S_k}/n\le 2p_k$ for $n$ sufficiently large.  From the bound on $C_{m}^{jl,(\bk)}$ in \eqref{eq:Cbdd}, we also have
\begin{align} 
\left| \sum_{||\bk||>N}\beta_l\alpha^S p_\bk\Xth^\bk C_{m}^{jl,(\bk)}(t,\alpha^S\Xth\bm{\p\psi}(\Xth)) \right | &\le  \beta_l L \sum_{||\bk|| > N}  ||\bk||^2p_\bk.  \label{Htail}
\end{align}
Then \eqref{eq:Deltatail} follows from \eqref{mmtail} and \eqref{Htail} together with \ref{A2}.

 Next we  show that $\sup_{0<t\le T}\left| \Delta^{SI}_{jl,(\bk)}(t) \right| \stackrel{P}{\to}0$ for any $\bk$.  
We write 
\begin{align}
| \Delta^{SI}_{jl,(\bk)}(t) |  &= \left| n^{-1}\sum_{i\in S_\bk} \left(\beta_lX_l^{SI,i} (X_j^{SS,i}-X_j^{SI,i})\right) - \beta_l\alpha^S p_\bk\Xth^\bk C_{m}^{jl,(\bk)}(t,\alpha^S\Xth\bm{\p\psi}(\Xth)) \right| \non \\
&\le  n^{-1} \beta_l \left|  \sum_{i \in S_\bk}  Q^{jl,(\bk)}_i -X_{S_k}C^{jl,(\bk)}_h(t) \right| \label{term1} \\
& \qquad +   n^{-1} \beta_l X^{S_\bk} \left|C^{jl,(\bk)}_h(t)-C_{m}^{jl,(\bk)}(t,n^{-1}\XSdot)\right| \label{term2} \\
& \qquad + \beta_l \left | n^{-1}X^{S_\bk}C_{m}^{jl,(\bk)}(t,n^{-1}\XSdot) - \alpha^Sp_\bk\Xth^\bk C_{m}^{jl,(\bk)}(t,n^{-1}\XSdot) \right|  \label{term3} \\
& \qquad + \beta_l \alpha^S p_\bk \Xth^\bk  \left| C_{m}^{jl,(\bk)}(t,n^{-1}\XSdot) - C_{m}^{jl,(\bk)}(t,\alpha^S\Xth\bm{\p\psi}(\Xth))\right| \label{term4}
\end{align}
and we will show that each of these terms tends to zero uniformly in probability.

By Remark \ref{rem:hg} and equation \eqref{eq:EhCh}, the process $M^{jl,(\bk)}_{h}(t)=\sum_{i \in S_\bk} Q^{jl,(\bk)}_i - X_{S_k}C^{jl,(\bk)}_h(t)$
is a zero-mean, piecewise-constant  c\`{a}dl\`{a}g martingale that jumps only if infection/recovery of a node of degree $\bk$ (or a neighbor of a node of degree $\bk$) occurs or activation/drop of a $j$-edge or $l$-edge belonging to a node of degree $\bk$ occurs.  Recall that, for each layer, either activation or drops are possible, not both.  Consider events impacting a node $u \in S_\bk$.  For infection or recovery events, of which there are at most $2(1+k_j+k_l)$ corresponding to infection and recovery of $u$ itself or one of its $j$- or $l$-neighbors), the jump size is at most $k_jk_l$.  For deactivation and activation events of an $l$ or $j$-edge, of which there are at most $k_l+k_j$ affecting $u$, the jump size is also at most $k_jk_l$.  Recall that the number of nodes of degree $\bk$ is approximately $np_\bk$ for large $n$.  The quadratic variation of $M^{jl,(\bk)}_{h}(t)$ is the sum of its  squared jumps (see, e.g.,\cite{Andersson2000} Chapter 9) and, thus, satisfies 
 \[[M^{jl,(\bk)}_{h}](t)=\sum_{s\le t} (\delta M^{jl,(\bk)}_{h}(s))^2 \le 2np_\bk(1+k_j+k_l)(k_jk_l)^2+np_\bk(k_l+k_j)(k_jk_l)^2 \le L_4||\bk||^5 n\]
for $0 < t \le T <\infty$ and some $L_4>0$.
  Since $E[M^{jl,(\bk)}_{h}](t)=E(M^{jl,(\bk)}_{h}(t))^2=O(n)$, Doob's martingale inequality implies
 \[\sup_{0<t\le T} \left| n^{-1} M^{jl,(\bk)}_{h}(t)\right|\stackrel{P}{\to}0,\]
 i.e. the term in \eqref{term1} tends to zero uniformly in probability.

In consideration of the term in \eqref{term2}, we note that $n^{-1}X_{S_k} \le 1$ and $C^{jl,(\bk)}_h(t) = C_{m}^{jl,(\bk)}(t,n^{-1}\XSdot)$ for $j \neq l$.  For the case $l=j$, we have
\begin{gather*}\sup_{0<t\le T}\left|C^{jj,(\bk)}_h(t)-C_{m}^{jj,(\bk)}(t,n^{-1}\XSdot)\right| = \sup_{0<t\le T}\left| k_j(k_j-1)\dfrac{X_j^{SI}(t)}{X_j^{S\b}(t)}\left(\dfrac{X_j^{S\b}(t)+X_j^{SS}(t)-X_j^{SI}(t)}{X_j^{S\b}(t)(X_j^{S\b}(t)-1)} \right) \right| \\ \le \dfrac{2L'||\bk||^2}{X_j^{S\b}(T)-1}\end{gather*}
for some $L'>0$ and since $X_j^{S\b}(t)$ is non-increasing on $[0,T]$.  Thus, the term in \eqref{term2} tends to zero uniformly in probability by \ref{A1}. 

For the term in  \eqref{term3}, we observe
\begin{gather*} \sup_{0<t\le T}\left | n^{-1}X^{S_\bk}C_{m}^{jl,(\bk)}(t,n^{-1}\XSdot) - \alpha^Sp_\bk\Xth^\bk C_{m}^{jl,(\bk)}(t,n^{-1}\XSdot) \right| \\  \le L||k||^2\sup_{0<t\le T} \left|n^{-1}X^{S_\bk}-\alpha^Sp_\bk\Xth^\bk\right| \stackrel{P}{\rightarrow}0 \end{gather*}
by the bound on $C_{m}^{jl,(\bk)}$ in \eqref{eq:Cbdd} and Lemma \ref{lem:XSdot}\ref{lem:fixk}.

Finally, since $C_{m}^{jl,(\bk)}(t,\bm{z}(t))$ is Lipschitz continuous in $\bm{z}$, we have
 \[\sup_{0<t\le T}\left|C_{m}^{jl,(\bk)}\left(t,n^{-1}\XSdot\right) -C_{m}^{jl,(\bk)}(t,\alpha^S\Xth\bm{\p\psi}(\Xth)) \right| \le L'\sup_{0<t\le T} \norml n^{-1}\XSdot-\alpha^S\Xth\bm{\p\psi}(\Xth)\normr \]
 for some $L'>0$ and so the term in  \eqref{term4} tends to zero uniformly in probability by Lemma \ref{lem:XSdot}\ref{lem:thpth}.  Therefore, recalling also \eqref{eq:Deltatail} we conclude that \eqref{eq:Djl} holds and hence
 \eqref{eq:DeltaSItozero} follows. 
 \end{proof}

We may now complete the derivation of Theorem \ref{thm:lln} via Gronwall's  inequality (see, e.g., \cite{Andersson2000}).

\paragraph*{\em Proof of Theorem \ref{thm:lln}.}
Recall the definition of $\bm{\Delta} = (\bm{\Delta}^\X,\bm{\Delta}^\Xth)$ from equations \eqref{eq:Delta} and \eqref{eq:Deltath}. Note that, by equations  \eqref{eq:X} and \eqref{eq:intth},
 \[(\X(t)/n,\Xth(t))=(\X(0)/n,\Xth(0))+\int_0^t \bm{\H}(\X(s)/n,\Xth(s)) ds+ \bm{\E}(t)\]
where 
\[\bm{\E}(t)=(n^{-1}\bm{M}(t),0) + \int_0^t \bm{\Delta}(s)ds,\]
 where $\bm{M}(t)=(M^S,M^I,\bm{M^{SI}},\bm{M^{\SI}},\bm{M^{SS}},\bm{M^{\SS}})(t)$. Note that  each coordinate of  $\bm{M}(t)$ 
 is a  pure jump,  c\`{a}dl\`{a}g, zero mean, martingale process. Consider the process  $M_j^{SI}$ which, by equation \eqref{eq:X},  jumps  only  if  infection of a node, recovery of a node or a $j$-edge drop/activation  occurs at time $s$.  Recall that, for each $j$, either activations or drops are possible, not both.  Consider events corresponding to a node of degree $\bk$.  For infection and recovery events  the  jump size,  $\delta M_j^{SI}(s)$, is not greater than that node's $j$-degree, and  for activation and deactivation events, of which there are at most $2k_j$ affecting that node, the jump size is one.  Since the number of nodes of degree $\bk$ is approximately  $n p_\bk$ for large $n$, the corresponding quadratic variation process satisfies 
\[E[M_j^{SI}](t)=E\sum_{s\le t}(\delta M_j^{SI}(s))^2\le  2n\sum_\bk  k_j^2 p_\bk + 2n\sum_{\bk}k_jp_\bk \le 4n\sum_{\bk} k_j^2p_\bk =O(n)\]
 by \ref{A2}. Consequently, Doob's martingale inequality implies 
$\sup_{0<t\le T} ||n^{-1}\bm{M^{SI}}(t) ||\stackrel{P}{\to}0$. A similar argument applies also to $\bm{M^{\SI}}$, $\bm{M^{SS}}$, and $\bm{M^{\SS}}$ as well as $M_S$ and $M_I$, both of which make only unit jumps. Since by Lemma~\ref{lem:Delta} we have $\sup_{0<t\le T}||\bm{\Delta}(t)||\stackrel{P}{\to} 0$, this implies $\sup_{0<t\le T}||\bm{\E}(t)||\stackrel{P}{\to} 0$.

 Note that  $\bm{\H}$ is  a (vector valued) Lipschitz continuous  function on its domain, which we can take to be $[0,1]^2 \times ([\bm{\xi},2\bm{\p\psi}(\one)]^r)^4 \times [\bm{\xi},\one]^r$ by Remark \ref{rem:bdd}. Together with Gronwall's inequality, this implies
\begin{align*}
&\sup_{0<s\le t} ||(\X/n,\Xth)(t)-(\bx,\th)(t)||\\ 
&\le ||(\X/n,\Xth)(0)-(\bx,\th)(0)||+ \sup_{0<s\le t}\norml \int_{0}^t [\bm{\H}((\X/n,\Xth)(s))-\bm{\H}((\bx,\th)(s))]ds \normr + \sup_{0<s\le t} ||\bm{\E}(s) ||\\
&\le ||(\X/n,\Xth)(0)-(\bx,\th)(0)||+ L\int_{0}^t \sup_{0<s\le t}||(\X/n,\Xth)(s)-(\bx,\th)(s) ||ds  + \sup_{0<s\le t} ||\bm{\E}(s) ||\\
&\le \left(||(\X/n,\Xth)(0)-(\bx,\th)(0)|| + \sup_{0<s\le t} ||\bm{\E}(s) ||\right) e^{Lt}\stackrel{P}{\to} 0,
 \end{align*} 
for some $L>0$, since the first term in the parenthesis tends to zero by \ref{A3} and \eqref{A4}. The assertion follows when we take $t=T<\infty$.

\section{Equivalence with edge-based multiple modes of transmission model (Proof of Corollary \ref{cor:EBequiv}) } 	\label{app:equiv}

We provide in this section the proof of Corollary \ref{cor:EBequiv}, i.e. equivalence of the system \eqref{eq:D} with the edge-based model with multiple modes of transmission given by system \eqref{eq:EB}.  We will use bar notation to denote the variables of the edge-based model.   

First, we observe that $S(0) = \alpha_S\psi(\th(0)) = \alpha_S$ and 
\[\dfrac{d(\alpha_S\psi(\th))}{dt} = \alpha_S\sum_j \p_j\psi(\th)\dfrac{d\theta_j}{dt} = -\sum_j \beta_j[SI]_j,\]
i.e. $\alpha_S\psi(\th)$ satisfies the same differential equation as $S$.  Therefore, by the uniqueness of the ODE solution,
\begin{equation}
S = \alpha_S\psi(\th). \label{eq:DS}
\end{equation}

Secondly, we claim that
\begin{equation}
[SS]_j = \alpha_S^2\dfrac{(\p_j\psi(\th))^2}{\p_j\psi(\one)}. \label{eq:DSS}
\end{equation}
Indeed, $[SS]_j(0) = \alpha_S^2\p_j\psi(\one)$ and 
\begin{align*}
\dfrac{d}{dt}\left(\alpha_S^2\dfrac{(\p_j\psi(\th))^2}{\p_j\psi(\one)}\right) &= \dfrac{2\alpha_S^2}{\p_j\psi(\one)} \p_j\psi(\th)\dfrac{d}{dt}\p_j\psi(\th) \\
&= -\dfrac{2\alpha_S\p_j\psi(\th)}{\p_j\psi(\one)} \sum_l\beta_l\p^2_{jl}\psi(\th)\dfrac{[SI]_l}{\p_l\psi(\th)} \\
&=-2\sum_l\beta_l\bar{\kappa}_{jl}(\th)\dfrac{[SI]_l}{S}\left(\alpha_S^2\dfrac{(\p_j\psi(\th))^2}{\p_j\psi(\one)}\right),
\end{align*}
which is  the same differential equation as that satisfied by $[SS]_j$, which proves \eqref{eq:DSS}.

We define
\[\phi^{I,j} = \theta_j -\alpha_S\dfrac{\p_j\psi(\th)}{\p_j\psi(1)} - \frac{\gamma}{\beta_j}(1-\theta_j) - \alpha_R, \qquad j = 1,\hdots,r.\]
Then, $\phi^{I,j}(0) = 1-\alpha_S-\alpha_R =\alpha_I$.  Moreover,
\begin{align*}
\dfrac{d\phi^{I,j}}{dt} &= \dfrac{d\theta_j}{dt} - \dfrac{\alpha_S}{\p_j\psi(1)}\sum_l \p^2_{jl}\psi(\th)\dfrac{d\theta_l}{dt} + \dfrac{\gamma}{\beta_j}\dfrac{d\theta_j}{dt} \\
&= \dfrac{1}{\p_j\psi(\one)}\sum_l \beta_l\dfrac{\p^2_{jl}\psi(\th)}{\p_l\psi(\th)}[SI]_l - (\beta_j+\gamma)\dfrac{[SI]_j}{\alpha_S\p_j\psi(\th)} \\
&= \dfrac{1}{\p_j\psi(\one)}\sum_l \beta_l\bar{\kappa}_{jl}(\th)\dfrac{\p_j\psi(\th)}{\psi(\th)}[SI]_l - (\beta_j+\gamma)\dfrac{[SI]_j}{\alpha_S\p_j\psi(\th)}.
\end{align*}
We now show that $\frac{[SI]_j}{\alpha_S\p_j\psi(\th)}$ satisfies the same differential equation:
\begin{align*}
\dfrac{d}{dt}\left(\dfrac{[SI]_j}{\alpha_S\p_j\psi(\th)}\right) &= \dfrac{\dfrac{d[SI]_j}{dt}}{\alpha_S\p_j\psi(\th)} - \dfrac{[SI]_j\dfrac{d}{dt}(\alpha_S\p_j\psi(\th))}{(\alpha_S\p_j\psi(\th))^2} \\
&= \sum_l \beta_l\bar{\kappa}_{jl}(\th)\dfrac{[SI]_l}{\alpha_SS\p_j\psi(\th)}([SS]_j-[SI]_j) - (\beta_j+\gamma)\dfrac{[SI]_j}{\alpha_S\p_j\psi(\th)} + \dfrac{[SI]_j}{\alpha_S(\p_j\psi(\th))^2}\sum_l\beta_l\dfrac{\p^2_{jl}\psi(\th)[SI]_l}{\alpha_S\p_l\psi(\th)} \\
&= \sum_l \beta_l\bar{\kappa}_{jl}(\th)\dfrac{[SI]_l}{\alpha_SS\p_j\psi(\th)}([SS]_j-[SI]_j) - (\beta_j+\gamma)\dfrac{[SI]_j}{\alpha_S\p_j\psi(\th)} + \dfrac{[SI]_j}{\alpha_S\p_j\psi(\th)}\sum_l\beta_l\bar{\kappa}_{jl}(\th)\dfrac{[SI]_l}{S} \\
&= \dfrac{1}{\p_j\psi(\one)}\sum_l \beta_l\bar{\kappa}_{jl}(\th)\dfrac{\p_j\psi(\th)}{\psi(\th)}[SI]_l - (\beta_j+\gamma)\dfrac{[SI]_j}{\alpha_S\p_j\psi(\th)} 
\end{align*}
where we have used \eqref{eq:DS} and \eqref{eq:DSS}.  Since $\frac{[SI]_j(0)}{\alpha_S\p_j\psi(\th(0))} =\alpha_I$ it follows from the uniqueness of the ODE solution that 
\[\phi^{I,j} = \dfrac{[SI]_j}{\alpha_S\p_j\psi(\th)}.\]
Hence, $\dfrac{d\theta_j}{dt} = -\beta_j\phi^{I,j}$, i.e.
\[\dfrac{d\theta_j}{dt} = -\beta_j\theta_j + \alpha_S\beta_j\dfrac{\p_j\psi(\th)}{\p_j\psi(1)} + \gamma(1-\theta_j) + \beta_j\alpha_R,\]
which is the same differential equation as that for $\bar{\theta}_j$ in \eqref{eq:EB}.  Furthermore, $\theta_j(1) =\bar{\theta}_j(1)$ which implies that 
\[\theta_j = \bar{\theta}_j.\]
It subsequently follows from \eqref{eq:EB} and equation \eqref{eq:DS} that $S$, $I$ and $R$ are also equivalent for the two models.

\section{Derivations of $\mu_i^S$ and $\mu_{i|j}^{ex|SI}$} 
	\label{app:kappabar}
	\def\bs{{\bm s}}
We provide here the derivations for $\mu^S_i$ and $\mu^{ex|SI}_{i|j}$ as given in equation \eqref{eq:muS}.  Recall that $\mu^S_i$ is the average $i$-degree of a susceptible node.  By equation \eqref{eq:EZi}, the probability that a node $u$ is susceptible and of degree $\bk$ is given by $P(u \in S_\bk(t)) = n^{-1}S(0)p_\bk\Xth^\bk(t)$.  We can then calculate $\mu^S_i$ as follows:
\begin{align*}
\mu^S_i(t) &= \sum_\bk k_iP(u \in S_\bk(t) | u \in S(t)) = \dfrac{\sum_\bk k_iP(u \in S_\bk(t))}{\sum_\bs P(u \in S_\bs(t))} = \dfrac{\sum_\bk k_ip_\bk\Xth^\bk(t)}{\sum_\bs p_\bs\Xth^\bk(t)} = \dfrac{\xth_i(t)\p_i\psi(\Xth(t))}{\psi(\Xth(t))}.
\end{align*}

Recall that $\mu^{ex|SI}_{i|j}$ is the average excess $i$-degree of a susceptible node chosen randomly as a $j$-neighbor of an infectious node.  That is, we randomly select a $j$-edge between a susceptible node and an infectious node, and we calculate the excess $i$ degree of the susceptible node.  Let $E_j^{SI}$ denote the set of $j$-edges between susceptible nodes and infectious nodes and let $E_j^{S_{\bk}I}$ denote the set of $j$-edges between susceptible nodes of degree $\bk$ and infectious nodes, so that $E^{SI}_j = \cup_\bk E_j^{S_{\bk}I}$.  Recall that $X^{SI,u}_j$ is the number of infectious $j$-neighbors of a susceptible node $u$.  Also recall from Remark \ref{rem:hg} that, given $u \in S_\bk$, the neighborhood of $u$ has a hypergeometric distribution and $E_h[X_j^{SI,u}] = k_jX_j^{SI}/X_j^{S\b}$. We first assume $i \neq j$ and calculate
\begin{align*}
\mu^{ex|SI}_{i|j} &= \sum_\bk k_i P\left(e \in E^{S_{\bk}I}_j | e \in E_j^{SI} \right) \\
& = \sum_{\bk} k_i  \dfrac{\sum_{l=0}^{k_j} lP\left(u \in S_{\bk}, X^{SI,u}_j = l\right)}{\sum_{\bs} \sum_{l=0}^{s_j} lP\left(u \in S_{\bs}, X^{SI,u}_j = l\right)}  \\
&= \sum_{\bk} k_i  \dfrac{P(u \in S_{\bk})\sum_{l=0}^{k_j} lP\left(X^{SI,u}_j = l | u \in S_{\bk} \right)}{\sum_{\bs} P (u \in S_{\bs})\sum_{l=0}^{s_j} lP\left(X^{SI,u}_j = l | u \in S_{\bs} \right) } \\
&= \sum_{\bk} k_i  \dfrac{P(u \in S_{\bk})k_jX^{SI}_j/X^{S\b}_j}{\sum_{\bs} P (u \in S_{\bs})s_jX^{SI}_j/X^{S\b}_j} \\
&= \dfrac{\sum_{\bk} k_i k_j p_\bk\Xth^\bk}{\sum_{\bs} s_jp_\bs\Xth^\bs} \\
&= \dfrac{\xth_i\p^2_{ij}\psi(\Xth)}{\p_j\psi(\Xth)}.
\end{align*}
If $j = i$, we likewise get
\begin{align*}
\mu^{ex|SI}_{i|i}  &= \dfrac{\sum_\bk (k_i-1) k_i p_\bk\Xth^\bk}{\sum_\bk k_ip_\bk\Xth^\bk} = \dfrac{\xth_i\p^2_{ii}\psi(\Xth)}{\p_i\psi(\Xth)}.
\end{align*}
Therefore, equation \eqref{eq:muS} follows.

Alternatively, we recall the equivalent model with dynamic graph construction mentioned in Section \ref{sec:proof} and consider a $j$-half edge of an infectious node that is forced to pair with a $j$-half edge of a susceptible node, which we denote by $e$, at time $t$.  Let $E^S_j$ denote the set of $j$-half edges belonging to susceptible nodes and let $E^{S_\bk}_j$ denote those belonging to susceptible nodes of degree $\bk$.  We first assume $i \neq j$ and calculate $\mu^{ex|SI}_{i|j}$ as follows:
\begin{align*}
\mu^{ex|SI}_{i|j} = \sum_\bk k_i P\left(e \in E^{S_{\bk}}_j | e \in E_j^S \right) =\sum_\bk k_i  \dfrac{k_j P(u \in S_\bk)}{\sum_\bs s_jP(u \in S_\bs)} = \dfrac{\sum_\bk k_i k_j p_\bk\Xth^\bk}{\sum_\bk k_jp_\bk\Xth^\bk} = \dfrac{\xth_i\p^2_{ij}\psi(\Xth)}{\p_j\psi(\Xth)}.
\end{align*}
The case $i=j$ follows likewise as above.
 	
\section{Calculations for community-healthcare model} 
	\label{app:CH}

\subsection{Basic reproduction number, $\R$} \label{app:R0}
We use the next-generation matrix method (and corresponding notation) from \cite{vdd2002} to calculate $\R$.  We consider $I$, $[SI]_C$, $[SI]_H$, and $[\SI]_H$ to be the infective compartments ($m=4$).  We consider the disease-free equilibrium for the system \eqref{eq:CH} given by $x_0=(1,0,0,0,0,\mu_C,\mu_H)$.  The matrix corresponding to terms for new infections, $F$, is then given by 
\[F = \left[\begin{array}{cccc}
0 & \beta_C & \beta_H & 0 \\
0 & \beta_C\kappa_{CC}\mu_C & \beta_H\kappa_{CH}\mu_C & 0 \\
0 & 0 & 0 & 0 \\
0 & \beta_C\kappa_{CH}\mu_H & \beta_H\kappa_{HH}\mu_H & 0 
\end{array}
\right]\]
and the matrix corresponding to terms from all other transitions, $V$, is
\[V = \left[\begin{array}{cccc}
\gamma & 0 & 0 & 0 \\
0 & \beta_C+\gamma+\delta & 0 & 0 \\
0 & 0 & \beta_H+\gamma & -\eta \\
0 & 0 & 0  & \eta+\gamma
\end{array}
\right].\]
Therefore,
\[V^{-1} = \left[\begin{array}{cccc}
1/\gamma & 0 & 0 & 0 \\
0 & 1/(\beta_C+\gamma+\delta) & 0 & 0 \\
0 & 0 & 1/(\beta_H+\gamma) & \eta/((\beta_H+\gamma)(\eta+\gamma)) \\
0 & 0 & 0  & 1/(\eta+\gamma)
\end{array}
\right],\]
and the next-generation matrix is given by
\[FV^{-1} = \left[\begin{array}{cccc}
0 & \beta_C/(\beta_C+\gamma+\delta) & \beta_H/(\beta_H+\gamma) & \beta_H\eta/((\beta_H+\gamma)(\eta+\gamma)) \\
0 & \beta_C\kappa_{CC}\mu_C/(\beta_C+\gamma+\delta) & \beta_H\kappa_{CH}\mu_C/(\beta_H+\gamma) & \beta_H\kappa_{CH}\mu_C\eta/((\beta_H+\gamma)(\eta+\gamma)) \\
0 & 0 & 0 & 0 \\
0 & \beta_C\kappa_{CH}\mu_H/(\beta_C+\gamma+\delta) & \beta_H\kappa_{HH}\mu_H/(\beta_H+\gamma) & \beta_H\kappa_{HH}\mu_H\eta/((\beta_H+\gamma)(\eta+\gamma)) 
\end{array}
\right].\]
Then, $\R$ is the spectral radius of the next-generation matrix, i.e. the largest absolute value of an eigenvalue.  We see that $FV^{-1}$ has two zero eigenvalues, and the other two eigenvalues are determined by the characteristic polynomial:
\[p(\lambda) = \lambda^2 - \left(\dfrac{\beta_C\kappa_{CC}\mu_C}{\beta_C+\gamma+\delta}+\dfrac{\beta_H\kappa_{HH}\mu_H\eta}{(\beta_H+\gamma)(\eta+\gamma)}\right)\lambda + \dfrac{(\kappa_{CC}\kappa_{HH}-\kappa_{CH}^2)\beta_C\beta_H\mu_C\mu_H\eta}{(\beta_C+\gamma+\delta)(\beta_H+\gamma)(\eta+\gamma)}.\]
Let \[R_C = \dfrac{\beta_C\kappa_{CC}\mu_C}{\beta_C+\gamma+\delta}, \qquad \text{and} \qquad R_H = \dfrac{\beta_H\kappa_{HH}\mu_H\eta}{(\beta_H+\gamma)(\eta+\gamma)}.\]
We solve $p(\lambda)=0$ to determine
\[\R = \dfrac{1}{2}\left(R_C+R_H\right) + \dfrac{1}{2}\sqrt{\left(R_C+R_H\right)^2+4\dfrac{\beta_C\mu_C\beta_H\mu_H\eta}{(\beta_C+\gamma+\delta)(\beta_H+\gamma)(\eta+\gamma)}(\kappa_{CH}^2-\kappa_{CC}\kappa_{HH})}. \]

\subsection{Independent Poissons case} \label{app:final}

We derive an invariant and a final size relation in the case of independent layers with Poisson distributions (i.e. $\kappa_{CC}=\kappa_{HH}=\kappa_{CH} = 1$).  First, we observe that, in this case,
\[\dfrac{d[SS]_C}{dt} = 2\dfrac{[SS]_C}{S}\dfrac{dS}{dt}\]
which implies that $[SS]_C = \mu_CS^2$.  Likewise, $[\SS]_H = \mu_HS^2$.
 
We then transform the system \eqref{eq:CH} using
\[Q_i = \dfrac{[SI]_i}{S}, \quad \tilde{Q}_i = \dfrac{[\SI]_i}{S}, \qquad \text{for $i = C,H$},\]
which gives:
\begin{equation}
\begin{aligned}
\dfrac{dS}{dt} &= -\beta_CQ_CS-\beta_HQ_HS \\
\dfrac{dI}{dt} &= \beta_CQ_CS+ \beta_HQ_HS -\gamma I \\
\dfrac{dQ_C}{dt} &= \mu_C(\beta_CQ_CS+\beta_HQ_HS)-(\beta_C+\gamma+\delta)Q_C \\
\dfrac{dQ_H}{dt} &= -(\beta_H+\gamma)Q_H+\eta \tilde{Q}_H \\
\dfrac{d\tilde{Q}_H}{dt} &= \mu_H(\beta_CQ_CS+\beta_HQ_HS)-(\eta+\gamma)\tilde{Q}_H. \label{eq:Pois}
\end{aligned}
\end{equation}

Observe that \eqref{eq:Pois} fits the form of Arino et al. \cite{Arino2007}.  Matching their notation, we ignore the decoupled $I$ equation and let the infected compartments be $x=(Q_C,Q_H,\tilde{Q}_H)^T$, the susceptible compartment $y=S$, the transmission row vector $\beta b = (\beta_C, \beta_H, 0)$, $\Pi = (\mu_C, 0, \mu_H)^T$, and 
\[V = \left [ \begin{array}{ccc}
\beta_C+\gamma+\delta & 0 & 0 \\
0 & \beta_H+\gamma  & -\eta \\
0 & 0 & \eta + \gamma
\end{array} \right ].\]
The system can then be written as
\begin{align}
\dfrac{dx}{dt} &= \Pi y \beta b x - Vx  \label{eq:x}\\
\dfrac{dy}{dt} &= -y \beta b x. \label{eq:y}
\end{align}
Recall that $y(0)= S(0)= \alpha_S$, $x(0) = (\mu_C\alpha_I,0,\mu_H\alpha_I)^T$ and from equation \eqref{eq:R0} we have
\[\R = \dfrac{\beta_C\mu_C}{\beta_C+\gamma+\delta} + \dfrac{\beta_H\mu_H\eta}{(\beta_H+\gamma)(\eta+\gamma)} = \beta bV^{-1}\Pi = \beta bV^{-1}\frac{x(0)}{\alpha_I}.
\]

Observe that
\[\Pi \dfrac{dy}{dt} + \dfrac{dx}{dt} = -Vx,\]
which implies that 
\[\int_0^t x(s)ds = V^{-1}\Pi(y(0)-y(t))+V^{-1}(x(0)-x(t)).\]
Then, from integrating equation \eqref{eq:y}, we have
\begin{align*}
\log\left(\dfrac{y(t)}{y(0)}\right) &= -\beta b \int_0^t x(s)ds \\
& = -\beta b V^{-1}\Pi(y(0)-y(t))-\beta bV^{-1}(x(0)-x(t)) \\
&= -\R(\alpha_S-y(t)) -\R\alpha_I +\beta bV^{-1}x(t).
\end{align*}
Transforming back to our original variables, this gives the invariant \eqref{CHinv}:
\begin{equation*}
\log\left(\frac{S}{\alpha_S}\right) = -\R\left(\alpha_S+\alpha_I-S\right)+\dfrac{\beta_C}{\beta_C+\gamma+\delta}\dfrac{[SI]_C}{S} +\dfrac{\beta_H}{\beta_H+\gamma}\dfrac{[SI]_H}{S} +\dfrac{\beta_H\eta}{(\beta_H+\gamma)(\eta+\gamma)}\dfrac{[\SI]_H}{S}.
\end{equation*}

Let $S_\infty$ denote the fraction of the population that escapes infection. Taking the limit $t \to \infty$ gives the final size relation \eqref{CHfinal}:
\begin{equation*}
\log\left(\frac{S_\infty}{\alpha_S}\right) = -\R\left(\alpha_S+\alpha_I-S_\infty\right).
\end{equation*}

\subsection{Model reduction} \label{app:inv}

Let $\sigma$ and $\lambda$ be as defined by equation \eqref{eq:sigma}.  We will derive relation \eqref{inv1} to demonstrate our method of finding invariants. Suppose we have a function $f$ which satisfies
\[f(S,[SS]_C,[\SS]_H)=0.\]
Let $A = \p f/\p S$, $B = \p f/\p[SS]_C$, and $C =\p f/\p [\SS]_H$.
Then,
\[\dfrac{df}{dt} = A\dfrac{dS}{dt}+B\dfrac{d[SS]_C}{dt}+C\dfrac{d[\SS]_H}{dt} = 0,\]
which by system \eqref{eq:CH} implies
\begin{equation*}
 \beta_C[SI]_C\left(A + 2B\kappa_{CC}\dfrac{[SS]_C}{S}+2C\kappa_{CH}\dfrac{[\SS]_H}{S}\right) + \beta_H[SI]_H\left(A+ 2B\kappa_{CH}\dfrac{[SS]_C}{S}+2C\kappa_{HH}\dfrac{[\SS]_H}{S}\right) =0.
\end{equation*}

Thus, $df/dt=0$ if the following system of equations is satisfied:
\begin{align*}
A + 2B\kappa_{CC}\dfrac{[SS]_C}{S}+2C\kappa_{CH}\dfrac{[\SS]_H}{S} = 0,\\
A+ 2B\kappa_{CH}\dfrac{[SS]_C}{S}+2C\kappa_{HH}\dfrac{[\SS]_H}{S} =0.
\end{align*}
For a given $C$, the solution to this system is:
\begin{align*}
A &= -2C\lambda \dfrac{[\SS]_H}{S} \\
B &= -C\sigma \dfrac{[\SS]_H}{[SS]_C}. 
\end{align*}
We choose $C= \partial f/ \partial [\SS]_H = 1$ which implies $f = [\SS]_H + g([SS]_C,S) = 0$ for some function $g$.  Therefore, 
\[\partdot{g}{S} = \partdot{f}{S} = A = -2\lambda \dfrac{[\SS]_H}{S} = 2\lambda \dfrac{g}{S}\]
from which separation of variables gives $g = S^{2\lambda}h([SS]_C)$ for some function $h$.  Hence,
\[S^{2\lambda}\dfrac{dh}{d[SS]_C} = \partdot{g}{[SS]_C} = \partdot{f}{[SS]_C} = B = \sigma \dfrac{g}{[SS]_C} = \sigma S^{2\lambda} \dfrac{h}{[SS]_C}\]
which gives 
\[\dfrac{dh}{h} = \sigma \dfrac{d[SS]_C}{[SS]_C} \]
from which separation of variables gives $h = K[SS]_C^\sigma$ for any constant $K$.  Therefore, $df/dt = 0$ is satisfied by any $f$ of the form:
\[f = [\SS]_H + g([SS]_C,S) = [\SS]_H + S^{2\lambda}h([SS]_C) = [\SS]_H + K[SS]_C^\sigma S^{2\lambda}.\]
Substituting the initial conditions \eqref{CHic3} into $f=0$  gives 
\[K=-\dfrac{[\SS]_H(0)}{([SS]_C(0))^\sigma (S(0))^{2\lambda}} = -\dfrac{\mu_H\alpha_S^{2(1-\sigma-\lambda)}}{\mu_C^\sigma}.\]
Hence, the invariant is
\begin{equation*}
[\SS]_H - \dfrac{\mu_H\alpha_S^{2(1-\sigma-\lambda)}}{\mu_C^\sigma}[SS]_C^\sigma S^{2\lambda} = 0,
\end{equation*}
i.e. relation \eqref{inv1}.

The derivations of relations \eqref{inv2} and \eqref{inv3}, the additional invariants in the independent layers case, are similar.  For \eqref{inv2}, we look for an invariant of the form $f(S,[SS]_C,[SI]_C)=0$.  Likewise, for \eqref{inv3} we consider an invariant of the form $f(S,[\SS]_H,[SI]_H,[\SS]_H)=0$.  The analysis follow analogously to above.

\begin{acknowledgements}
We thank Mason Porter and KaYin Leung for their helpful comments during manuscript preparation.  We also thank the Mathematical Biosciences Institute at The Ohio State University for its assistance in providing us with space and the necessary computational resources.  
\end{acknowledgements}

\bibliographystyle{spmpsci}      \bibliography{ebola_refs2}   

\begin{thebibliography}{100}
\providecommand{\url}[1]{{#1}}
\providecommand{\urlprefix}{URL }
\expandafter\ifx\csname urlstyle\endcsname\relax
  \providecommand{\doi}[1]{DOI~\discretionary{}{}{}#1}\else
  \providecommand{\doi}{DOI~\discretionary{}{}{}\begingroup
  \urlstyle{rm}\Url}\fi

\bibitem{Altmann1995}
Altmann, M.: Susceptible-infected-removed epidemic models with dynamic
  partnerships.
\newblock Journal of Mathematical Biology \textbf{33}(6), 661--675 (1995)

\bibitem{Altmann1998}
Altmann, M.: The deterministic limit of infectious disease models with dynamic
  partners.
\newblock Mathematical Biosciences \textbf{150}(2), 153--175 (1998)

\bibitem{andersson1998limit}
Andersson, H.: Limit theorems for a random graph epidemic model.
\newblock Annals of Applied Probability pp. 1331--1349 (1998)

\bibitem{Andersson2000}
Andersson, H., Britton, T.: Stochastic epidemic models and their statistical
  analysis, vol.~4.
\newblock Springer New York (2000)

\bibitem{Arino2007}
Arino, J., Brauer, F., Van Den~Driessche, P., Watmough, J., Wu, J.: A final
  size relation for epidemic models.
\newblock Mathematical Biosciences and Engineering \textbf{4}(2), 159 (2007)

\bibitem{balcan2009}
Balcan, D., Colizza, V., Gon\c{c}alves, B., Hu, H., Ramasco, J.J., Vespignani,
  A.: Multiscale mobility networks and the spatial spreading of infectious
  diseases.
\newblock Proceedings of the National Academy of Sciences \textbf{106}(51),
  21,484--21,489 (2009)

\bibitem{Ball2008}
Ball, F., Neal, P.: Network epidemic models with two levels of mixing.
\newblock Mathematical biosciences \textbf{212}(1), 69--87 (2008)

\bibitem{Bansal2007}
Bansal, S., Grenfell, B.T., Meyers, L.A.: When individual behaviour matters:
  homogeneous and network models in epidemiology.
\newblock Journal of the Royal Society Interface \textbf{4}(16), 879--891
  (2007)

\bibitem{Bansal2010}
Bansal, S., Read, J., Pourbohloul, B., Meyers, L.A.: The dynamic nature of
  contact networks in infectious disease epidemiology.
\newblock Journal of Biological Dynamics \textbf{4}(5), 478--489 (2010)

\bibitem{barbour2013approximating}
Barbour, A.D., Reinert, G., et~al.: Approximating the epidemic curve.
\newblock Electron. J. Probab \textbf{18}(54), 1--30 (2013)

\bibitem{barthelemy2004}
Barth{\'e}lemy, M., Barrat, A., {Pastor-Satorras}, R., Vespignani, A.: Velocity
  and hierarchical spread of epidemic outbreaks in scale-free networks.
\newblock Physical Review Letters \textbf{92}(178701) (2004)

\bibitem{barthelemy2005}
Barth{\'e}lemy, M., Barrat, A., {Pastor-Satorras}, R., Vespignani, A.:
  Dynamical patterns of epidemic outbreaks in complex heterogeneous networks.
\newblock Journal of Theoretical Biology \textbf{235}, 275--288 (2005)

\bibitem{Battiston2014}
Battiston, F., Nicosia, V., Latora, V.: Structural measures for multiplex
  networks.
\newblock Physical Review E \textbf{89}(3), 032,804 (2014)

\bibitem{bengtsson2015}
Bengtsson, L., Gaudart, J., Lu, X., Moore, S., Wetter, E., Sallah, K.,
  Rebaudet, S., Piarroux, R.: Using mobile phone data to predict the spatial
  spread of cholera.
\newblock Scientific Reports \textbf{5}, 10.1038/srep08,923 (2015)

\bibitem{bohman2012sir}
Bohman, T., Picollelli, M.: Sir epidemics on random graphs with a fixed degree
  sequence.
\newblock Random Structures \& Algorithms \textbf{41}(2), 179--214 (2012)

\bibitem{brockmann2010}
Brockmann, D.: Human mobility and spatial disease dynamics.
\newblock In: H.G. Schuster (ed.) Reviews of nonlinear dynamics and complexity,
  volume 2. Wiley-VCH Verlag GmbH Co, KGaA, Weinheim, Germany (2010)

\bibitem{brockmann2013}
Brockmann, D., Helbing, D.: The hidden geometry of complex, network-driven
  contagion phenomena.
\newblock Science \textbf{342}(6164), 1337--1342 (2013)

\bibitem{Buono2014}
Buono, C., Alvarez-Zuzek, L.G., Macri, P.A., Braunstein, L.A.: Epidemics in
  partially overlapped multiplex networks.
\newblock {PLoS ONE} \textbf{9}(3), 5 (2014)

\bibitem{Cardillo2013}
Cardillo, A., Zanin, M., G{\'o}mez-Garde{\~n}es, J., Romance, M., del Amo,
  A.J.G., Boccaletti, S.: Modeling the multi-layer nature of the {E}uropean
  {A}ir {T}ransport {N}etwork: Resilience and passengers re-scheduling under
  random failures.
\newblock The European Physical Journal Special Topics \textbf{215}(1), 23--33
  (2013)

\bibitem{Chowell2014}
Chowell, G., Nishiura, H.: Transmission dynamics and control of ebola virus
  disease (evd): a review.
\newblock BMC medicine \textbf{12}(1), 196 (2014)

\bibitem{colizza2006}
Colizza, V., Barrat, A., Barth'{e}lemy, M., Vespignani, A.: The role of the
  airline transportation network in the prediction and predictability of global
  epidemics.
\newblock Proceedings of the National Academy of Sciences USA \textbf{103},
  2015--2020 (2006)

\bibitem{coltart2015}
Coltart, C.E.M., Johnson, A.M., Whitty, C.J.M.: Role of healthcare workers in
  early epidemic spread of {Ebola}: policy implications of prophylactic
  compared to reactive vaccination policy in outbreak prevention and control.
\newblock BMC Medicine \textbf{13}, 271 (2015)

\bibitem{De-Domenico2013}
De~Domenico, M., Sol{\'e}-Ribalta, A., Cozzo, E., Kivel{\"a}, M., Moreno, Y.,
  Porter, M.A., G{\'o}mez, S., Arenas, A.: Mathematical formulation of
  multilayer networks.
\newblock Physical Review X \textbf{3}(4), 041,022 (2013)

\bibitem{Decreusefond2012}
Decreusefond, L., Dhersin, J.S., Moyal, P., Tran, V.C., et~al.: Large graph
  limit for an {SIR} process in random network with heterogeneous connectivity.
\newblock The Annals of Applied Probability \textbf{22}(2), 541--575 (2012)

\bibitem{dodds2004}
Dodds, P.S., Watts, D.J.: Universal behavior in a generalized model of
  contagion.
\newblock Physical Review Letters \textbf{92}(21), 218,701(1) -- 218,201(4)
  (2004)

\bibitem{dowell1999}
Dowell, S.F., Mukunu, R., Ksiazek, T.G., Khan, A.S., Rollin, P.E., Peters, C.:
  Transmission of ebola hemorrhagic fever: a study of risk factors in family
  members, kikwit, democratic republic of the congo, 1995.
\newblock Journal of Infectious Diseases \textbf{179}(Supplement 1), S87--S91
  (1999)

\bibitem{Drake2015}
Drake, J.M., Kaul, R., Alexander, L.W., O'Regan, S.M., Kramer, A.M., Pulliam,
  J.T., Ferrari, M.J., Park, A.W.: Ebola cases and health system demand in
  {L}iberia.
\newblock PLoS Biol \textbf{13}(1), e1002,056 (2015)

\bibitem{Eames2002}
Eames, K.T., Keeling, M.J.: Modeling dynamic and network heterogeneities in the
  spread of sexually transmitted diseases.
\newblock Proceedings of the National Academy of Sciences \textbf{99}(20),
  13,330--13,335 (2002)

\bibitem{Eames2004}
Eames, K.T., Keeling, M.J.: Monogamous networks and the spread of sexually
  transmitted diseases.
\newblock Mathematical biosciences \textbf{189}(2), 115--130 (2004)

\bibitem{easley2010}
Easley, D., Kleinberg, J.: Networks, crowds, and markets: reasoning about a
  highly connected world.
\newblock Cambridge University Press (2010)

\bibitem{Epstein2008}
Epstein, J.M., Parker, J., Cummings, D., Hammond, R.A.: Coupled contagion
  dynamics of fear and disease: mathematical and computational explorations.
\newblock PLoS One \textbf{3}(12), e3955 (2008)

\bibitem{Fisman2014}
Fisman, D., Khoo, E., Tuite, A.: Early epidemic dynamics of the {W}est
  {A}frican 2014 {E}bola outbreak: estimates derived with a simple
  two-parameter model.
\newblock PLoS Currents \textbf{6} (2014)

\bibitem{Funk2010b}
Funk, S., Gilad, E., Jansen, V.: Endemic disease, awareness, and local
  behavioural response.
\newblock Journal of Theoretical Biology \textbf{264}(2), 501--509 (2010)

\bibitem{Funk2009}
Funk, S., Gilad, E., Watkins, C., Jansen, V.A.: The spread of awareness and its
  impact on epidemic outbreaks.
\newblock Proceedings of the National Academy of Sciences \textbf{106}(16),
  6872--6877 (2009)

\bibitem{Funk2010}
Funk, S., Jansen, V.A.: Interacting epidemics on overlay networks.
\newblock Physical Review E \textbf{81}(3), 036,118 (2010)

\bibitem{goeijenbier2014}
Goeijenbier, M., {van Kampen}, J.J.A., Reusken, C.B.E.M., Koopmans, M.P.G.,
  {van Gorp}, E.C.M.: Ebola virus disease: a review on epidemiology, symptoms,
  treatment and pathogenesis.
\newblock The Journal of Medicine \textbf{72}(9), 442--448 (2014)

\bibitem{Gomes2014}
Gomes, M.F., y~Piontti, A.P., Rossi, L., Chao, D., Longini, I., Halloran, M.E.,
  Vespignani, A.: Assessing the international spreading risk associated with
  the 2014 {W}est {A}frican {E}bola outbreak.
\newblock PLoS Currents \textbf{6} (2014)

\bibitem{Granell2013}
Granell, C., G{\'o}mez, S., Arenas, A.: Dynamical interplay between awareness
  and epidemic spreading in multiplex networks.
\newblock Physical Review Letters \textbf{111}(12), 128,701 (2013)

\bibitem{Grassberger1983}
Grassberger, P.: On the critical behavior of the general epidemic process and
  dynamical percolation.
\newblock Mathematical Biosciences \textbf{63}(2), 157--172 (1983)

\bibitem{Gross2006}
Gross, T., D'Lima, C.J.D., Blasius, B.: Epidemic dynamics on an adaptive
  network.
\newblock Physical review letters \textbf{96}(20), 208,701 (2006)

\bibitem{hewlett2003}
Hewlett, B.S., Amola, R.P.: Cultural contexts of {Ebola} in northern {Uganda}.
\newblock Emerging Infectious Diseases \textbf{9}, 1242--1248 (2003)

\bibitem{House2011}
House, T., Keeling, M.J.: Insights from unifying modern approximations to
  infections on networks.
\newblock Journal of The Royal Society Interface \textbf{8}(54), 67--73 (2011)

\bibitem{Janson2014}
Janson, S., Luczak, M., Windridge, P.: Law of large numbers for the {SIR}
  epidemic on a random graph with given degrees.
\newblock Random Structures \& Algorithms \textbf{45}(4), 724--761 (2014)

\bibitem{Jo2006}
Jo, H.H., Baek, S.K., Moon, H.T.: Immunization dynamics on a two-layer network
  model.
\newblock Physica A: Statistical Mechanics and its Applications
  \textbf{361}(2), 534--542 (2006)

\bibitem{Keeling1999}
Keeling, M.: Correlation equations for endemic diseases: externally imposed and
  internally generated heterogeneity.
\newblock Proceedings of the Royal Society of London B: Biological Sciences
  \textbf{266}(1422), 953--960 (1999)

\bibitem{keeling1999b}
Keeling, M.J.: The effects of local spatial structure on epidemiological
  invasions.
\newblock Proceedings of the Royal Society of London B: Biological Sciences
  \textbf{266}(1421), 859--867 (1999)

\bibitem{kermack1927}
Kermack, W.O., McKendrick, A.G.: A contribution to the mathematical theory of
  epidemics.
\newblock Proceedings of the Royal Society of London Series A \textbf{115},
  700--721 (1927)

\bibitem{kivela2014}
Kivel{\"a}, M., Arenas, A., Barthelemy, M., Gleeson, J.P., Moreno, Y., Porter,
  M.A.: Multilayer networks.
\newblock Journal of Complex Networks \textbf{2}, 203--271 (2014)

\bibitem{kratz2015}
Kratz, T., Roddy, P., Oloma, A.T., Jeffs, B., Ciruelo, D.P., {de la Rosa}, O.,
  Borchert, M.: {Ebola Virus Disease} outbreak in {Isiro, Democratic Republic
  of the Congo}, 2012: signs and symptoms, management and outcomes.
\newblock {PLoS ONE} \textbf{10}(6), e0129,333 (2015)

\bibitem{kurtz1970solutions}
Kurtz, T.G.: Solutions of ordinary differential equations as limits of pure
  jump markov processes.
\newblock Journal of Applied Probability \textbf{7}(1), 49--58 (1970)

\bibitem{Legrand2007}
Legrand, J., Grais, R., Boelle, P., Valleron, A., Flahault, A.: Understanding
  the dynamics of ebola epidemics.
\newblock Epidemiology and infection \textbf{135}(04), 610--621 (2007)

\bibitem{Lekone2006}
Lekone, P.E., Finkenst{\"a}dt, B.F.: Statistical inference in a stochastic
  epidemic {SEIR} model with control intervention: {E}bola as a case study.
\newblock Biometrics \textbf{62}(4), 1170--1177 (2006)

\bibitem{Lenhart2007}
Lenhart, S., Workman, J.T.: Optimal control applied to biological models.
\newblock CRC Press (2007)

\bibitem{leung2012}
Leung, K.Y., Kretzschmar, M., Diekmann, O.: Dynamic concurrent partnership
  networks incorporating demography.
\newblock Theoretical Population Biology \textbf{82}, 229--239 (2012)

\bibitem{leung2015}
Leung, K.Y., Kretzschmar, M., Diekmann, O.: {$SI$} infection on a dynamic
  partnership network: characterization of {$R_0$}.
\newblock Journal of Mathematical Biology \textbf{71}, 1--56 (2015)

\bibitem{Li2015}
Li, M., Ma, J., van~den Driessche, P.: Model for disease dynamics of a
  waterborne pathogen on a random network.
\newblock Journal of mathematical biology \textbf{71}(4), 961--977 (2015)

\bibitem{lindquist2011}
Lindquist, J., Ma, J., {van den Driessche}, P., Willeboordse, F.H.: Effective
  degree network disease models.
\newblock Journal of Mathematical Biology \textbf{62}(2), 143--164 (2011)

\bibitem{Ma2006}
Ma, J., Earn, D.J.: Generality of the final size formula for an epidemic of a
  newly invading infectious disease.
\newblock Bulletin of Mathematical Biology \textbf{68}(3), 679--702 (2006)

\bibitem{ma2013}
Ma, J., {van den Driessche}, P., Willeboordse, F.H.: Effective degree household
  network disease model.
\newblock Journal of Mathematical Biology \textbf{66}, 75--94 (2013)

\bibitem{maganga2014}
Maganga, G.D., Kapetshi, J., Berthet, N., Kebela~Ilunga, B., Kabange, F.,
  Mbala~Kingebeni, P., Mondonge, V., Muyembe, J.J.T., Bertherat, E., Briand,
  S., et~al.: Ebola virus disease in the {D}emocratic {R}epublic of {C}ongo.
\newblock New England Journal of Medicine \textbf{371}(22), 2083--2091 (2014)

\bibitem{matanock2014}
Matanock, A., Arwady, M.A., Ayscue, P., Forrester, J.D., Gaddis, B., Hunter,
  J.C., Monroe, B., Pillai, S.K., Reed, C., Schafer, I.J., Massaquoi, M., Dahn,
  B., {De Cock}, K.M.: {Ebola Virus Disease} cases among health care workers
  not working in {Ebola} treatment units - {Liberia, June-August}, 2014.
\newblock Morbidity and Mortality Weekly Report \textbf{63}(46), 1077--1081
  (2014)

\bibitem{May2001}
May, R.M., Lloyd, A.L.: Infection dynamics on scale-free networks.
\newblock Physical Review E \textbf{64}(6), 066,112 (2001)

\bibitem{Merler2015}
Merler, S., Ajelli, M., Fumanelli, L., Gomes, M.F., y~Piontti, A.P., Rossi, L.,
  Chao, D.L., Longini, I.M., Halloran, M.E., Vespignani, A.: Spatiotemporal
  spread of the 2014 outbreak of {E}bola virus disease in {L}iberia and the
  effectiveness of non-pharmaceutical interventions: a computational modelling
  analysis.
\newblock The Lancet Infectious Diseases \textbf{15}(2), 204--211 (2015)

\bibitem{Meyer1962}
Meyer, P., et~al.: A decomposition theorem for supermartingales.
\newblock Illinois Journal of Mathematics \textbf{6}(2), 193--205 (1962)

\bibitem{meyers2005}
Meyers, L.A., Pourbohloul, B., {Newman}, M.E., Skowronski, D.M., Brunham, R.C.:
  Network theory and {SARS}: predicting outbreak diversity.
\newblock Journal of Theoretical Biology \textbf{232}, 71--81 (2005)

\bibitem{Miller2011b}
Miller, J.C.: A note on a paper by {E}rik {V}olz: {SIR} dynamics in random
  networks.
\newblock Journal of mathematical biology \textbf{62}(3), 349--358 (2011)

\bibitem{Miller2014b}
Miller, J.C.: Epidemics on networks with large initial conditions or changing
  structure.
\newblock {PLoS ONE} \textbf{9}(7), e101,421 (2014)

\bibitem{Miller2014}
Miller, J.C., Kiss, I.Z.: Epidemic spread in networks: Existing methods and
  current challenges.
\newblock Mathematical Modelling of Natural Phenomena \textbf{9}(2), 4 (2014)

\bibitem{Miller2011}
Miller, J.C., Slim, A.C., Volz, E.M.: Edge-based compartmental modelling for
  infectious disease spread.
\newblock Journal of The Royal Society Interface p. rsif20110403 (2011)

\bibitem{Miller2013}
Miller, J.C., Volz, E.M.: Incorporating disease and population structure into
  models of {SIR} disease in contact networks.
\newblock {PLoS ONE} \textbf{8}(8), e69,162 (2013)

\bibitem{newman2010}
Newman, M.: Networks: an introduction.
\newblock Oxford University Press (2010)

\bibitem{newman2002}
{Newman}, M.E.: Spread of epidemic disease on networks.
\newblock Physical Review E  (2002)

\bibitem{Pandey2014}
Pandey, A., Atkins, K.E., Medlock, J., Wenzel, N., Townsend, J.P., Childs,
  J.E., Nyenswah, T.G., Ndeffo-Mbah, M.L., Galvani, A.P.: Strategies for
  containing {E}bola in {W}est {A}frica.
\newblock Science \textbf{346}(6212), 991--995 (2014)

\bibitem{pastor-satorras2001}
Pastor-Satorras, R., Vespignani, A.: Epidemic spreading in scale-free networks.
\newblock Physical Review Letters \textbf{86}(14), 3200--3203 (2001)

\bibitem{Pellis2015}
Pellis, L., Ball, F., Bansal, S., Eames, K., House, T., Isham, V., Trapman, P.:
  Eight challenges for network epidemic models.
\newblock Epidemics \textbf{10}, 58--62 (2015)

\bibitem{pellis2015exact}
Pellis, L., House, T., Keeling, M.J.: Exact and approximate moment closures for
  non-{M}arkovian network epidemics.
\newblock Journal of theoretical biology \textbf{382}, 160--177 (2015)

\bibitem{porter2009}
Porter, M.A., Onnela, J.P., Mucha, P.J.: Communities in networks.
\newblock Notices of the American Mathematical Society \textbf{56}(9),
  1082--1097, 1164--1166 (2009)

\bibitem{Rand1999}
Rand, D.: Correlation equations and pair approximations for spatial ecologies.
\newblock Advanced Ecological Theory: Principles and Applications \textbf{100}
  (1999)

\bibitem{Rivers2014}
Rivers, C.M., Lofgren, E.T., Marathe, M., Eubank, S., Lewis, B.L.: Modeling the
  impact of interventions on an epidemic of {E}bola in {S}ierra {L}eone and
  {L}iberia.
\newblock PLoS Currents \textbf{6} (2014)

\bibitem{roels1999}
Roels, T.H., Bloom, A.S., Buffington, J., Muhungu, G.L., {MacKenzie}, W.R.,
  Khan, A.S., Ndambi, R., Noah, D.L., Rolka, H.R., Peters, C.J., Ksiazek, T.G.:
  Ebola hemorrhagive fever, {Kikwit, Democratic Republic of the Congo, 1995}:
  risk factors for patients without a reported exposure.
\newblock Journal of Infectious Diseases \textbf{179}(Suppl 1), S92--S97 (1999)

\bibitem{Rombach2014}
Rombach, M.P., Porter, M.A., Fowler, J.H., Mucha, P.J.: Core-periphery
  structure in networks.
\newblock SIAM Journal on Applied Mathematics \textbf{74}(1), 167--190 (2014)

\bibitem{Sahneh2013}
Sahneh, F.D., Scoglio, C.: May the best meme win!: new exploration of
  competitive epidemic spreading over arbitrary multi-layer networks.
\newblock arXiv preprint arXiv:1308.4880  (2013)

\bibitem{Scarpino2014}
Scarpino, S.V., Iamarino, A., Wells, C., Yamin, D., Ndeffo-Mbah, M., Wenzel,
  N.S., Fox, S.J., Nyenswah, T., Altice, F.L., Galvani, A.P., et~al.:
  Epidemiological and viral genomic sequence analysis of the 2014 {E}bola
  outbreak reveals clustered transmission.
\newblock Clinical Infectious Diseases p. ciu1131 (2014)

\bibitem{Shai2012}
Shai, S., Dobson, S.: Effect of resource constraints on intersimilar coupled
  networks.
\newblock Physical Review E \textbf{86}(6), 066,120 (2012)

\bibitem{Shai2013}
Shai, S., Dobson, S.: Coupled adaptive complex networks.
\newblock Physical Review E \textbf{87}(4), 042,812 (2013)

\bibitem{Sharkey2008}
Sharkey, K.J.: Deterministic epidemiological models at the individual level.
\newblock Journal of Mathematical Biology \textbf{57}(3), 311--331 (2008)

\bibitem{sharkey2015exact}
Sharkey, K.J., Kiss, I.Z., Wilkinson, R.R., Simon, P.L.: Exact equations for
  {SIR} epidemics on tree graphs.
\newblock Bulletin of mathematical biology \textbf{77}(4), 614--645 (2015)

\bibitem{Shaw2008}
Shaw, L.B., Schwartz, I.B.: Fluctuating epidemics on adaptive networks.
\newblock Physical Review E \textbf{77}(6), 066,101 (2008)

\bibitem{tatem2009}
Tatem, A.J., Qiu, Y.L., Smith, D.L., Sabot O. amd~Ali, A.S., Moonen, B.: The
  use of mobile phone data for the estimation of the travel patterns and
  imported {{\it Plasmodium falciparum}} rates among {Zanzibar} residents.
\newblock Malaria Journal \textbf{8} (2009)

\bibitem{Tsanou2015}
Tsanou, B., Moremedi, G.M., Kaondera-Shava, R., Lubuma, J.M., Morris, N.: A
  simple mathematical model for {E}bola in {A}frica.
\newblock Biomath Communications \textbf{2}(1) (2015)

\bibitem{ugander2012}
Ugander, J., Backstrom, L., Marlow, C., Kleinberg, J.: Structural diversity in
  social contagion.
\newblock Proceedings of the National Academy of Sciences USA \textbf{109}(16),
  5962--5966 (2012)

\bibitem{Valdano2015}
Valdano, E., Ferreri, L., Poletto, C., Colizza, V.: Analytical computation of
  the epidemic threshold on temporal networks.
\newblock Physical Review X \textbf{5}(2), 021,005 (2015)

\bibitem{Valdano2015b}
Valdano, E., Poletto, C., Colizza, V.: Infection propagator approach to compute
  epidemic thresholds on temporal networks: impact of immunity and of limited
  temporal resolution.
\newblock The European Physical Journal B \textbf{88}(12), 1--11 (2015)

\bibitem{Valdez2015}
Valdez, L., R{\^e}go, H.H.A., Stanley, H., Braunstein, L.: Predicting the
  extinction of {E}bola spreading in {L}iberia due to mitigation strategies.
\newblock arXiv preprint arXiv:1502.01326  (2015)

\bibitem{vdd2002}
{van den Driessche}, P., Watmough, J.: Reproduction numbers and sub-threshold
  endemic equilibria for compartmental models of disease transmission.
\newblock Mathematical Biosciences \textbf{180}, 29--48 (2002)

\bibitem{VanDerHofstad2009}
Van Der~Hofstad, R.: Random graphs and complex networks.
\newblock Available on http://www. win. tue. nl/rhofstad/NotesRGCN. pdf p.~11
  (2009)

\bibitem{Volz2008}
Volz, E.: {SIR} dynamics in random networks with heterogeneous connectivity.
\newblock Journal of Mathematical Biology \textbf{56}(3), 293--310 (2008)

\bibitem{volz2007}
Volz, E., Meyers, L.A.: Susceptible--infected--recovered epidemics in dynamic
  contact networks.
\newblock Proceedings of the Royal Society of London B: Biological Sciences
  \textbf{274}(1628), 2925--2934 (2007)

\bibitem{Wang2011}
Wang, Y., Xiao, G.: Effects of interconnections on epidemics in network of
  networks.
\newblock In: Wireless Communications, Networking and Mobile Computing (WiCOM),
  2011 7th International Conference on, pp. 1--4. IEEE (2011)

\bibitem{watts1998}
Watts, D.J., Strogatz, S.H.: {Collective dynamics of `small-world' networks}.
\newblock {Nature} \textbf{{393}}({6684}), {440--442} (1998).
\newblock \doi{{10.1038/30918}}

\bibitem{Webb2014}
Webb, G., Browne, C., Huo, X., Seydi, O., Seydi, M., Magal, P.: A model of the
  2014 {E}bola epidemic in {W}est {A}frica with contact tracing.
\newblock PLoS Currents \textbf{7} (2014)

\bibitem{Wei2012}
Wei, X., Valler, N., Prakash, B.A., Neamtiu, I., Faloutsos, M., Faloutsos, C.:
  Competing memes propagation on networks: a case study of composite networks.
\newblock ACM SIGCOMM Computer Communication Review \textbf{42}(5), 5--12
  (2012)

\bibitem{Weitz2015}
Weitz, J.S., Dushoff, J.: Modeling post-death transmission of {E}bola:
  Challenges for inference and opportunities for control.
\newblock Scientific Reports \textbf{5} (2015)

\bibitem{Wells2015}
Wells, C., Yamin, D., Ndeffo-Mbah, M.L., Wenzel, N., Gaffney, S.G., Townsend,
  J.P., Meyers, L.A., Fallah, M., Nyenswah, T.G., Altice, F.L., Atkins, K.E.,
  Galvani, A.P.: Harnessing case isolation and ring vaccination to control
  {E}bola.
\newblock {PLoS Neglected Tropical Diseases}  (2015)

\bibitem{wesolowski2012}
Wesolowski, A., Eagle, N., Tatem, A.J., Smith, D.L., Noor, A.M., Snow, R.W.,
  Buckee, C.O.: Quantifying the impact of human mobility on malaria.
\newblock Science \textbf{338}(6104), 267--270 (2012)

\bibitem{wilkinson2009stochastic}
Wilkinson, D.J.: Stochastic modelling for quantitative description of
  heterogeneous biological systems.
\newblock Nature Reviews Genetics \textbf{10}(2), 122--133 (2009)

\bibitem{Yagan2012}
Ya{\u{g}}an, O., Gligor, V.: Analysis of complex contagions in random multiplex
  networks.
\newblock Physical Review E \textbf{86}(3), 036,103 (2012)

\bibitem{Yagan2013}
Yagan, O., Qian, D., Zhang, J., Cochran, D.: Conjoining speeds up information
  diffusion in overlaying social-physical networks.
\newblock Selected Areas in Communications, IEEE Journal on \textbf{31}(6),
  1038--1048 (2013)

\bibitem{Zanette2008}
Zanette, D.H., Risau-Gusm{\'a}n, S.: Infection spreading in a population with
  evolving contacts.
\newblock Journal of biological physics \textbf{34}(1-2), 135--148 (2008)

\bibitem{Zhao2014}
Zhao, D., Li, L., Li, S., Huo, Y., Yang, Y.: Identifying influential spreaders
  in interconnected networks.
\newblock Physica Scripta \textbf{89}(1), 015,203 (2014)

\end{thebibliography}

\end{document}